\documentclass[a4paper, 11pt]{article}
\RequirePackage{fullpage}

\usepackage{amsthm,amsmath,amssymb}
\usepackage{subcaption}
\usepackage[usenames,dvipsnames]{xcolor}
\usepackage[pdftex,breaklinks,colorlinks,
    citecolor={BlueViolet}, linkcolor={Blue},urlcolor=Maroon]{hyperref}
\usepackage{tikz,pgfkeys}
\usetikzlibrary{arrows.meta}
\usetikzlibrary{quotes}
\usetikzlibrary{decorations.pathreplacing}
\usepackage{graphicx}
\usepackage{charter,eulervm}%
\usepackage{multirow,booktabs,array}
\usepackage{tabularx,enumerate}
\usetikzlibrary{calc}
\usepackage[final,expansion=alltext,protrusion=true]{microtype}
\usepackage{tikz}

\theoremstyle{plain} 
\newtheorem{theorem}{Theorem}[section]
\newtheorem{lemma}[theorem]{Lemma}
\newtheorem{corollary}[theorem]{Corollary}
\newtheorem{proposition}[theorem]{Proposition}
\newtheorem{remark}[theorem]{Remark}

\tikzstyle{filled vertex}  = [{circle,blue,draw,fill=black!50,inner sep=1pt}]  

\newcommand{\codecomment}[1]{\textbackslash\!\!\textbackslash {\em #1}}
\newcommand{\lp}[1]{\ensuremath{{\mathtt{lp}(#1)}}}
\newcommand{\rp}[1]{\ensuremath{{\mathtt{rp}(#1)}}}

\title{Recognizing (Unit) Interval Graphs by \\Zigzag Graph Searches}
\author{Yixin Cao\thanks{Department of Computing, Hong Kong Polytechnic University, Hong Kong, China. \href{mailto:yixin.cao@polyu.edu.hk} {\tt yixin.cao@polyu.edu.hk}. 
    {Supported in part by the Hong Kong Research Grants Council (RGC) under grants 15201317 and 15226116, and the National Natural Science Foundation of China (NSFC) under grant 61972330.}
  }
}
\date{}

\begin{document}
\maketitle

\begin{abstract}
  Corneil, Olariu, and Stewart [SODA 1998; SIAM Journal on Discrete Mathematics 2009] presented a recognition algorithm for interval graphs by six graph searches.  Li and Wu [Discrete Mathematics \& Theoretical Computer Science 2014] simplified it to only four.  The great simplicity of the latter algorithm is however eclipsed by the complicated and long proofs.  The main purpose of this paper is to present a new and significantly short proof for Li and Wu's algorithm, as well as a simpler implementation. We also give a self-contained simpler interpretation of the recognition algorithm of Corneil [Discrete Applied Mathematics 2004] for unit interval graphs, based on three sweeps of graph searches.
Moreover, we show that two sweeps are already sufficient.  Toward the proofs of the main results, we make several new structural observations that might be of independent interests.
\end{abstract}

\thispagestyle{empty}
\setcounter{page}{0}

\newpage
\section{Introduction}\label{sec:intro}

If we take the 26 monarchy rulers of China, France, and United Kingdom (England and Scotland before 1707) during the period between 1661 and 1900, mark their reigns in the timeline, and make a graph to indicate the intersection relationships between their reigns, we end with the graph in Figure~\ref{fig:example}.  
Formally, a graph is an \emph{interval graph} if its vertices can be assigned to intervals on the real line such that there is an edge between two vertices if and only if their corresponding intervals intersect.  As in the opening example, interval graphs can be used to represent, among others, relations of a temporal nature.  Although it is very easy to draw a graph out of a set of intervals, the other direction, reconstructing an interval representation from a given interval graph, is a challenging task.  Indeed, it is already very nontrivial to decide whether a graph is an interval graph or not.
For example, it is not that obvious why the graph in Figure~\ref{fig:example} is an interval graph from the graph itself, without the background information.

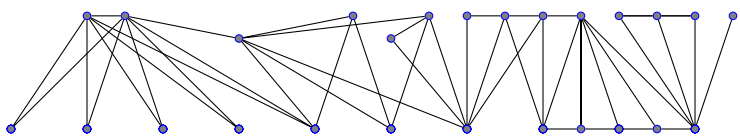
\begin{figure}[h]
  \centering\small
    \begin{tikzpicture}[every node/.style={filled vertex}, yscale=1.5]
      \foreach[count=\i] \j in {2.5, 5.5, 6.5, 7.5, 8.5, 9, 10, 10.5} 
        \node (c\i) at (\j, 2) {};
      \foreach[count=\i] \j/\y in {2/2, 4/1.8, 6/1.8, 7./2, 8/2, 8.5/1, 9.5/1, 9.5/2} {
        \node (f\i) at (\j, \y) {};
      \foreach[count=\i] \j in {1, ..., 10} 
        \node (e\j) at ({\i }, 1) {};
      }
      \foreach \i in {1, ..., 5} 
      \draw (c1) -- (e\i) -- (f1); %
      \draw (e5) -- (c2) -- (e6);
      \draw (e6) -- (c3) -- (e7);
      \draw (e5) -- (f2) -- (e6) (f2) -- (e7);

      \draw (c4) -- (e8) -- (c5);
      \draw (f4) -- (e7) -- (f5) (c4) -- (e7) -- (f3);
      \draw (e9) -- (c5) -- (f7);  
      
      \draw (c5) -- (e10) -- (c6) (c7) -- (e10) -- (c8) (f7) -- (e10) -- (f8);  
      \draw (c6) -- (f8) -- (c7);  
      \draw (f5) -- (e8) -- (f6); 
      \draw (f6) -- (e9) -- (f7); 
      \draw (f1) -- (c1) -- (f2) -- (c2)  (f2) -- (c3) -- (f3); %
      \draw (f4) -- (c4) -- (f5) -- (c5); %
      \draw (f6) -- (c5) -- (f6)  (c6) -- (f8) -- (c7); %
    \end{tikzpicture}
    \caption{The intersection graph of the reigns of 26 rulers.  Each vertex represents a monarchy ruler, and an edge indicates the overlap of the reigns of the two rulers involved.}
  \label{fig:example}
\end{figure}

Let us briefly relate the history of recognition algorithms for interval graphs.  Throughout we use $n$ to denote the number of vertices in the input graph.
An interval graph has at most $n$ maximal cliques, and Fulkerson and Gross~\cite{fulkerson-65-interval-graphs} observed that they can be arranged in a linear manner, called a \emph{clique path}, such that every vertex is contained in a consecutive set of them.
For example, the graph in Figure~\ref{fig:bfs-example} has five maximal cliques, which correspond to the five integers in the representation: In particular, the $i$th clique comprises vertices whose intervals containing point $i$.
All the maximal cliques of an interval graph can be found in linear time \cite{rose-76-vertex-elimination}.  (We know the answer is ``no'' if more than $n$ maximal cliques were found.)  Earlier recognition algorithms for interval graphs tried to orient the maximal cliques.  A na{\"i}ve implementation would take $O(n^3)$ time.
Booth and Lueker \cite{booth-76-pq-tree} invented a complicated data structure to capture the ordering of the maximal cliques, and showed how to use this data structure to recognize interval graphs in linear time.
Their approach was simplified by Korte and M{\"o}hring \cite{korte-89-recognizing-interval-graphs} and Hsu and McConnell~\cite{hsu-03-pc-trees}.
Another natural approach is to start from a clique tree of the input graph, defined similarly as a clique path and can be constructed in linear time \cite{rose-76-vertex-elimination}, and try to transform it into a path.
Both Hsu and Ma~\cite{hsu-99-recognizing-interval-graphs} and Habib et al.~\cite{habib-00-LBFS-and-partition-refinement} took this approach.

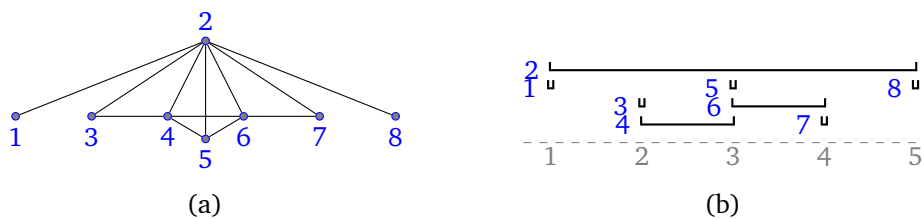
\begin{figure}[h]
  \centering\small
  \begin{subfigure}[b]{.35\linewidth}
    \centering
    \begin{tikzpicture}[every node/.style={filled vertex}, scale=1]
      \node ["$2$"] (v0) at (3.5, 2) {};
      \foreach[count=\i] \j in {1, 3, 4, 6, 7, 8} {
        \node ["$\j$" below] (v\j) at (\i, 1) {};
        \draw (v\j)-- (v0);
      }
      \node["$5$" below] (w) at (3.5, 0.7) {};
      \draw (v3)-- (v4)-- (v6)-- (v7) (v6)-- (w)--(v4) (v0) -- (w);
    \end{tikzpicture}
    \caption{}
  \end{subfigure}  
  \quad
  \begin{subfigure}[b]{.45\linewidth}
    \centering
    \begin{tikzpicture}[scale=1.2]
      \begin{scope}[every path/.style={{|[right]}-{|[left]}}]
        \foreach \x in {1, ..., 5} \node[gray] at (\x-1, -.15) {$\x$};
        \foreach[count=\i] \l/\r/\y/\c in {-.02/.02/3/, 0/4/4/,
          0.98/1.02/2/, 1/2/1/, 1.98/2.02/3/, 2/3/2/,
          2.98/3.02/1/, 3.98/4.02/3/}{
          \draw[\c, thick] (\l-.02, \y/5) node[left, blue] {$\i$} to (\r+.02, \y/5);
        }
      \end{scope}
      \draw[dashed, gray] (-.3, 0) -- (4.2, 0);
    \end{tikzpicture}
    \caption{}
  \end{subfigure}  
  \caption{(a) An interval graph and (b) its interval representation.}
  \label{fig:bfs-example}
\end{figure}

All the aforementioned algorithms are indirect, in the sense that they first check whether the input graph is chordal---a graph is \emph{chordal} if every induced cycle in it is a triangle---and continue only when the answer is yes.
Since any induced cycle longer than three cannot have an interval representation, all interval graphs are chordal.
Rose et al.~\cite{rose-76-vertex-elimination} has developed
a linear-time algorithm to decide whether a graph is chordal, and if yes, a clique tree can be built in the same time.
These maximal cliques are the starting point of aforementioned algorithms.
Since the structures of interval graphs are simpler than those of chordal graphs, one may consider the indirect approaches counterintuitive in a sense.

The algorithm of Rose et al.~\cite{rose-76-vertex-elimination} is based on a refined version of breadth-first search, called lexicographical breadth-first search (\textsc{lbfs} for short).
The output of \textsc{lbfs} is an ordering of the vertices, which satisfies certain property if and only if the graph is chordal.
Interestingly, the algorithms of Hsu and Ma~\cite{hsu-99-recognizing-interval-graphs} and Habib et al.~\cite{habib-00-LBFS-and-partition-refinement} heavily rely on \textsc{lbfs}.
Since a linear ordering of vertices is more perspicuous in interval graphs than in chordal graphs, it seems natural to also take the ordering approach to design a direct recognition algorithm for interval graphs.
Given an interval representation for an interval graph $G$, if we sort all the vertices by the left endpoints of their intervals in the representation, then for any three vertices in order, the adjacency of the first and the third vertices would force the adjacency of the first two vertices; see, e.g., the ordering $\langle 1, 2, \ldots, 8\rangle$ for the graph in Figure~\ref{fig:bfs-example}.
It turns out that
a graph $G$ is an interval graph if and only if $G$ admit an ordering with the this property \cite{ramalingam-88-domination-interval-graphs}, hence called an \emph{interval ordering}.

Since every interval ordering is an \textsc{lbfs} ordering, it seems very natural to use \textsc{lbfs} to produce an interval ordering of the input graph.  On the other hand, an \textsc{lbfs} ordering of an interval graph can be completely different from an interval ordering.  A necessary condition for an \textsc{lbfs} ordering to be an interval ordering is that it needs to start from certain vertices.
Simon~\cite{simon-91-interval} was probably the first who tried this approach.
After an initial \textsc{lbfs}, he conducted multiple sweeps of a variation of \textsc{lbfs}, called \textsc{lbfs}$^{+}$, each of which uses the outcome of the previous one to help break ties; in particular, it starts from the last vertex of the previous sweep.  Although fatally flawed, Simon's algorithm inspired fruitful exploration in this direction.  (According to \cite{corneil-09-lbfs-strucuture-and-interval-recognition}, Ma also made similar attempts.)  The real breakthrough was made by Corneil et al.~\cite{corneil-09-lbfs-strucuture-and-interval-recognition}, who observed that \textsc{lbfs}  and \textsc{lbfs}$^{+}$ may not be sufficient, and proposed one more variation of \textsc{lbfs}.  The development of this algorithm was not smooth: The four-sweep algorithm claimed in the conference version \cite{corneil-98-interval-recognition}, titled \textit{The ultimate interval graph recognition algorithm?}, turned out to be over-optimistic.  After significant reworking, they managed to make an algorithm that conducts six sweeps of \textsc{lbfs} and variations, of which the last sweep is very complicated.

Li and Wu~\cite{li-14-lbfs-interval-recognition} greatly simplified the algorithm of Corneil et al.~\cite{corneil-09-lbfs-strucuture-and-interval-recognition}.  They devised yet another variation of \textsc{lbfs} to catch the key concept of Corneil et al.~\cite{corneil-09-lbfs-strucuture-and-interval-recognition}.  Although their algorithm is greatly simple, the argument is excruciatingly complicated.  Even experts in interval graphs sometimes find difficulty in following the sequence of 32 propositions, with almost no explanation in between.  The main purpose of this paper is thus to give a short proof for this algorithm, with only cosmetic modifications to the algorithm itself.

An important subclass of interval graphs are unit interval graphs, those having interval representations in which all the intervals have the same length.  (Image a remote planet, where the president of each nation serves a six-year term, with no possibility of re-election.  Then the tenures of all the presidents form a unit interval graph.)  Unit interval graphs are also well studied, and there have been several recognition algorithms proposed in literature.  We also include a self-contained presentation of the recognition algorithm of Corneil~\cite{corneil-04-recognize-uig} for unit interval graphs, and our purpose is twofold.  First, we present both algorithms with a similar core idea, and thus the algorithm for unit interval graphs can be a warm-up of the main algorithm we are going to present.\footnote{This sentence is anachronistic.  The algorithm in~\cite{corneil-04-recognize-uig} was actually inspired by that for interval graphs~\cite{corneil-98-interval-recognition}.}   An oversimplistic summary of the algorithms for unit interval graphs and interval graphs is that they work by breaking, respectively, fake twins and fake modules.
Although the full details of the main result of this paper can be understood with Section~\ref{sec:uig} skipped, we strongly suggest the reader to take this detour.  Second, compared to the algorithm that is simple enough for a freshman, the analysis of the original paper \cite{corneil-04-recognize-uig} heavily relies on previous work.  Thus, a self-contained presentation might be worthwhile for its own.
(We have to, nevertheless, resist this attempt to give a self-contained presentation for Li and Wu's algorithm, because it would make the paper too long.  Since procedure \textsc{lbfs}$^{+}$ is well studied and well understood, we will cite results on them from literature, while proving everything else in the paper.)

During the presentation of the algorithms, we prove several new lemmas; as far as I can check, they have not been explicitly stated in literature.  Although none of the new lemmas can be called major, they do simplify and streamline a lot of important results.  For example, our characterizations of all (unit) interval orderings imply a simple two-sweep recognize algorithm for unit interval graphs.  We also prove the uniqueness of unit interval ordering of unit interval graphs with no true twins, discovered by Deng et al.~\cite{deng-96-proper-interval-and-cag}, and the uniqueness of clique paths of interval graphs with no nontrivial modules, discovered by Hsu~\cite{hsu-95-recognition-cag}.
We hope our new observations can shed more light on the algorithms and on (unit) interval graphs in general.

As a final remark, the success of conducting multiple runs of \textsc{lbfs} to recognize (unit) interval graphs had inspired a series of algorithms based on a similar approach.  We refer the interested reader to the survey of Corneil~\cite{corneil-04-survey-lbfs}.

\section{Preliminaries}\label{sec:lbfs}

All graphs discussed in this paper are undirected and simple.  The vertex set and edge set of a graph $G$ are denoted by, respectively, $V(G)$ and $E(G)$.
For a subset $U\subseteq V(G)$, denote by $G[U]$ the subgraph of $G$ induced by $U$, and by $G - U$ the subgraph $G[V(G)\setminus U]$.
The \emph{neighborhood} of a vertex $v$, denoted by $N(v)$, comprises vertices adjacent to $v$, i.e., $N(v) = \{ u \mid uv \in E(G) \}$, and the \emph{closed neighborhood} of $v$ is $N[v] = N(v) \cup \{ v \}$.
The \emph{closed neighborhood} and the \emph{neighborhood} of a set $X\subseteq V(G)$ of vertices are defined as $N[X] = \bigcup_{v \in X} N[v]$ and $N(X) =  N[X] \setminus X$, respectively.  A \emph{clique}\index{clique} is a set of pairwise adjacent vertices, and a clique is maximal if it is not proper subset of another clique.  A graph $G$ is \emph{complete} if $V(G)$ is a clique.  We say that a vertex $v$ is \emph{simplicial} if $N[v]$ is a clique; such a clique is necessarily maximal. 

A set of intervals representing an interval graph $G$ is called an \emph{interval representation} for $G$, where the interval for vertex $v$ is $I(v)$.  In this paper, all intervals are closed.  An example is given in Figure~\ref{fig:bfs-example}.  (Most authors would also require every interval to have a positive length.  However, it is more convenient for us to allow zero-length intervals.)
If for each of $n$ (not necessarily distinct) left endpoints of the $n$ intervals, we take the set of vertices whose intervals contain this point, then we end with $n$ cliques.  We leave it to the reader to verify that they include all the maximal cliques of $G$.  If we list the distinct maximal cliques from left to right, sorted by the endpoints that we use to define these cliques, then we can see that for any $v\in V(G)$, the maximal cliques containing $v$ appear consecutively.  We say that such a linear arrangement of maximal cliques a \emph{clique path} of $G$.  On the other hand, given a clique path $\langle K_{1}, K_{2}, \ldots, K_{\ell}\rangle$ for an interval graph $G$ with $\ell$ maximal cliques, for each vertex $v$ we can define an interval $[\lp{v}, \rp{v}]$, where $\lp{v}$ and $\rp{v}$ are the indices of the first and, respectively, last maximal cliques containing $v$.  One may easily see that they define an interval representation for $G$; see, e.g., Figure~\ref{fig:bfs-example}.  Therefore, a graph $G$ is an interval graph if and only if $G$ has a clique path.
As a consequence of this construction, every interval graph admits an interval representation in which all the endpoints are integers between $1$ and $n$.
One may note that an interval $I(v)$ obtained in this way has length zero if and only if $v$ is simplicial.

Yet another characterization of interval graphs is through vertex orderings.
An \emph{ordering} $\sigma$ of the vertex set of a graph $G$ is a bijection from $V(G) \rightarrow \{1, \ldots, n\}$.
If we scan the intervals in a given interval representation by their left endpoints, in non-decreasing order, we obtain an ordering of the vertices.  Without loss of generality, let it be $\langle v_1$, $v_2$, $\ldots$, $v_n\rangle$.  If for some $i$ and $k$ with $1\le i < k\le n$, vertices $v_i$ and $v_k$ are adjacent, then the right endpoint of $I(v_i)$ is larger than the left endpoint of $I(v_k)$.  By the ordering, for every $j$ with $i < j < k$, the left endpoint of $I(v_j)$ is contained in $I(v_i)$.  In other words, for any triple of integers $i, j, k$ with $1\le i <j <k\le n$, if $v_i v_k\in E(G)$ then $v_i v_j \in E(G)$ as well.
An ordering of $V(G)$ having this property is called an \emph{interval ordering} of $G$ \cite{ramalingam-88-domination-interval-graphs}.
On the other hand, it is also easy to derive an interval representation from an interval ordering $\sigma$, by setting $I(v_i) = \left[ i, j \right]$, where $v_j$ is the last neighbor of $v_i$ in $\sigma$.  Therefore, interval representations, clique paths, and interval orderings (equipped with the index of the last neighbor of each vertex) are essentially equivalent, and they can be transformed into each other without explicitly building the graph.  The direct transformation between a clique path and an interval ordering, to be introduced in Section~\ref{sec:algorithm}, will be used to argue the correctness of the algorithm.
However, one should be careful that both interval representations and clique paths have reflection symmetry, while the reversal of an interval ordering is not an interval ordering in general.  The interval ordering corresponding to the mirror image of an interval representation is the ordering of right endpoints, in non-increasing order; e.g., $\langle 2, 8, 6, 7, 4, 5, 3, 1\rangle$ for the graph in Figure~\ref{fig:bfs-example}.

Given an ordering $\sigma$ of $V(G)$, we can verify whether $\sigma$ is an interval ordering in linear time as follows.  First, we renumber the vertices such that $\sigma(v_i) = i$.  Second, sort the adjacency list for each vertex in the decreasing order.\footnote{As a standard algorithm exercise, this can be done without calling any sorting algorithm.  For example, we create $n$ new linked lists, each for a vertex, and all initially empty.  Then for $i = 1, \ldots, n$, we scan the neighbors of $v_i$, and add $i$ to the lists corresponding to neighbors of $v_i$.  After all the $n$ vertices are scanned, the $n$ lists contain the information we desire.  Recall that the insertion to a linked list is done at the front.}
For example, if $G$ is the graph in Figure~\ref{fig:bfs-example}, and $\sigma = \langle 1, 2, \ldots, 8\rangle$, then the lists are $[2]$; $[8, 7, \ldots, 3]$; $[4, 2]$; $[6, 5, 3, 2]$; $[6, 4, 2]$; $[7, 5, 4, 2]$; $[6, 2]$; and $[2]$.
To finish the task, it suffices to check for all $i = 1, \ldots, n - 1$, the $i$th list starts from $[f(i), f(i)-1, \ldots, i + 1]$, where $f(i)$ is the first number in the list; i.e., $v_{f(i)}$ is the last neighbor of $v_{i}$ in $\sigma$.
If $G$ is not an interval graph, then no ordering of $V(G)$ can be an interval ordering.  Thus, any ordering of $V(G)$ will fail our test.
For our purpose, therefore, it suffices to develop a procedure that finds an interval ordering if given an interval graph.  In our presentation we will focus on the behavior of the procedure on an interval graph.

\medskip

Breadth-first search (\textsc{bfs}) is arguably the simplest graph algorithm, and a backbone version, for a connected graph, can be described as follows.
\begin{itemize}
\item Initialize a queue whose only element is a starting vertex of the graph.
  \vspace*{-4pt}
\item While the queue is nonempty, dequeue a vertex $v$ from the queue, and enqueue all the neighbors of $v$ that have not already been visited (i.e., enqueued by earlier steps).
\end{itemize}
In \textsc{bfs}, there is no specific order on visiting the unvisited neighbors of $v$.  In the very extreme case, the starting vertex is \emph{universal}, i.e., adjacent to all other vertices, and the other vertices can be visited in any order.
(The standard output of \textsc{bfs} is a tree representing the discovery relation, in which the order of siblings is immaterial.)
Since we look for an ordering of the vertices, we may want a more refined control on the algorithm.
Whenever there is a set $S$ of vertices to be visited, we may conduct a \textsc{bfs} of the subgraph $G[S]$ and visit them in order.  If we impose this restriction recursively, then we end with \emph{lexicographic breadth-first search} (\textsc{lbfs}).
To see the difference between \textsc{bfs} and \textsc{lbfs}, let us consider the graph in Figure~\ref{fig:bfs-example}. In a \textsc{bfs} from vertex $1$, vertices $3$ through $8$ can be in any order, where no \textsc{lbfs} of the graph can end with vertex $4$, $5$, or $6$.  It is worth noting that an \textsc{lbfs} can still have ties, which we break arbitrarily.

The name \textsc{lbfs} comes from the following more standard description in Figure~\ref{fig:alg-lbfs}.\footnote{Our definition of lexicographic ordering is slightly unnatural.  As explained in the next paragraph, the original purpose of using \textsc{lbfs} was to find a perfect elimination ordering for a chordal graph.  In that setting, the labels in an \textsc{lbfs} are assigned from $n$ to $1$, instead of $1$ to $n$ as ours.  Also note that some authors present \textsc{lbfs} with a specific starting vertex, similar as \textsc{bfs} above.  We omit it because it is not necessary.}
Let $L_1$ and $L_2$ be two different subsets of $\{1, 2, \ldots, n\}$. 
We say that $L_1$ is \emph{lexicographically larger} than $L_2$ if the minimum element in $(L_1\setminus L_2) \cup (L_2\setminus L_1)$ belongs to $L_1$; e.g., $\{1, 2\}$ is lexicographically larger than both $\{1\}$ and $\{1, 3, 4\}$.  

\begin{figure}[h!]
  \centering 
  \begin{tikzpicture}
    \path (0,0) node[text width=.7\textwidth, inner sep=10pt] (a) {
      \begin{minipage}[t!]{\textwidth}

        \begin{tabbing}
          AAA\=Aaa\=aaa\=Aaa\=MMMMMAAAAAAAAAAAA\=A \kill
          1.\> {\bf for each} $v\in V(G)$ {\bf do}
          \\
          1.1.\>\>$\mathrm{label}(v) \leftarrow \emptyset$;
          \\
          2.\> {\bf for} $i = 1, \ldots, n$ {\bf do}
          \\
          2.1.\>\> $S\leftarrow$ unvisited vertices with the lexicographically largest label;
          \\
          2.2.\>\> $v\leftarrow$ any vertex in $S$;
          \\
          2.3.\>\> $\sigma(v)\leftarrow i$;
          \\
          2.4.\>\> \textbf{for each} unvisited neighbor of $v$ \textbf{do} 
          \\
          \>\>\> add $i$ to $\mathrm{label}(v)$;
          \\
          3.\> \textbf{return} $\sigma$.
        \end{tabbing}
      \end{minipage}
    };
    \draw[draw=gray!60] (a.north west) -- (a.north east) (a.south west) -- (a.south east);
  \end{tikzpicture}
  \caption{The procedure \textsc{lbfs}.}
  \label{fig:alg-lbfs}
\end{figure}

\begin{proposition}\label{prop:interval-is-lbfs}
  Let $G$ be an interval graph.  Any interval ordering $\sigma$ of $G$ is an \textsc{lbfs} ordering of $G$.
\end{proposition}
\begin{proof}
  We may assume without loss of generality that $\sigma(v_i) = i$ for all $i = 1, \ldots, n$.
  For any three numbers $i$, $p$, and $q$ with $i< p < q$, if $v_{i}v_{q}\in E(G)$, then $v_{i}v_{p}\in E(G)$ as well.  Therefore, $\{v_1, \ldots, v_{p-1}\}\cap N(v_q)\subseteq \{v_1, \ldots, v_{p-1}\}\cap N(v_p)$, and
after visiting $\{v_1, \ldots, v_{p-1}\}$, the label of $v_p$ is no smaller than that of $v_q$ for any $q > p$.
\end{proof}

Rose et al.~\cite{rose-76-vertex-elimination} devised \textsc{lbfs} for the purpose of recognizing chordal graphs.  They showed that the last vertex of an \textsc{lbfs} ordering $\sigma$ of a chordal graph $G$ is always simplicial.  In an \textsc{lbfs}, the decision on which vertex to visit next solely depends on the visited ones, and thus for $i = 1, \ldots, n$, the sub-ordering of the first $i$ vertices in $\sigma$ is an \textsc{lbfs} ordering of the subgraph induced by these vertices.  
Therefore, if we remove vertices in the reversal order of the \textsc{lbfs} ordering of a chordal graph, then every vertex is simplicial at the moment it is removed.  Such an ordering of removing vertices is called a \emph{perfect elimination ordering}, and it exists if and only if the graph is chordal \cite{dirac-61-chordal-graphs}.  
We use $u <_\sigma v$ to denote that $\sigma(u) < \sigma(v)$.
A \emph{sub-ordering} of an ordering $\sigma$ restricted to $S\subseteq V(G)$, denoted by $\sigma|_{S}$, is the ordering from $S \rightarrow \{1, \ldots, |S|\}$ such that $u<_{\sigma|_{S}} v$ if and only if $u<_{\sigma} v$ for all $u, v \in S$.
Note that a simplicial vertex $v$ of a graph $G$ is simplicial in $G[U]$ for any $U$ containing $v$.
\begin{theorem}[The perfect elimination theorem \cite{rose-76-vertex-elimination}]
  \label{thm:peo}
  Let $G$ be a chordal graph, and let $\sigma$ be an \textsc{lbfs} ordering of $G$.  For any subset $S\subseteq V(G)$, the last vertex of $\sigma|_{S}$ is simplicial in $G[S]$.
\end{theorem}

We say that a vertex $v$ of a graph $G$ is an (\textsc{lbfs}) \emph{end vertex} of $G$ if there exists an \textsc{lbfs} ordering $\sigma$ of $G$ such that $\sigma(v) = n$.  The key observation in \cite{rose-76-vertex-elimination} can be restated as that all end vertices of a chordal graph $G$ are simplicial in $G$.  The other direction of the perfect elimination theorem is not true in general; e.g., vertex $5$ in Figure~\ref{fig:bfs-example} is simplicial but cannot be the last vertex of any \textsc{lbfs} ordering of the graph.
The original success of \textsc{lbfs} thus crucially hinges on the end vertices.
On an interval graph $G$, 
a stronger property of end vertices was observed by Simon~\cite[Lemma 16]{simon-91-interval}.  This property is also implicit in \cite{corneil-09-lbfs-strucuture-and-interval-recognition}, and another proof for it can be found in \cite[Corollary 4.23]{li-14-lbfs-interval-recognition}.

\begin{theorem}[The end vertex theorem \cite{simon-91-interval, li-14-lbfs-interval-recognition}]
  \label{lem:end-vertex}
  Let  $G$ be an interval graph.  For any end vertex $v$ of $G$, there exists a clique path such that $v$ is a simplicial vertex of the first maximal clique.
\end{theorem}

The converse of Theorem~\ref{lem:end-vertex} is also true, and it is implied by the following characterization of interval orderings.
Let $G$ be an interval graph with $\ell$ maximal cliques, and let $\langle K_1, K_2, \ldots, K_{\ell}\rangle$ be a clique path of $G$.  We say an ordering $\sigma$ of $V(G)$ is \emph{consistent with} the clique path if $u<_{\sigma} v$ for any pair of vertices $u$ and $v$ with $\lp{u} < \lp{v}$; in other words, $\sigma$ can be represented as $\langle K_1, K_2\setminus K_1, \ldots, K_{\ell}\setminus K_{\ell - 1} \rangle$, where vertices in each set are in any order.
\begin{theorem}\label{thm:ordering-path}
  Let $G$ be an interval graph.  An ordering $\sigma$ of $V(G)$ is an interval ordering if and only if $\sigma$ is consistent with some clique path of $G$.
\end{theorem}
\begin{proof}
  The sufficiency is obvious, and hence we focus on necessity.
  We prove it by induction on $n$.  It is vacuously true for the base case, when $G$ contains a single vertex.  Now suppose that the claim is true for all graphs of order $n-1$.
  We may assume without loss of generality that $\sigma(v_i) = i$ for $i = 1, \ldots, n$.
  By Proposition~\ref{prop:interval-is-lbfs}, $\sigma$ is an \textsc{lbfs} ordering of $G$, and hence by the perfect elimination theorem (Theorem~\ref{thm:peo}), $v_{n}$ is simplicial in $G$, and $v_{n-1}$ is simplicial in $G - \{v_n\}$.
  By the induction hypothesis, there is a clique path of $G - \{v_n\}$ that is consistent with the ordering $v_1, \ldots, v_{n-1}$; denote by $K$ the last maximal clique of this clique path.
  
  Suppose first that there is a true twin of $v_n$ (i.e., another simplicial vertex of $G$ in $N[v_{n}]$).
  In this case, the simplicial vertex $v_{n-1}$ in $G - \{v_n\}$ is also simplicial in $G$.  
  We argue that $v_{n-1}\in N(v_{n})$: Otherwise, there exists $i < n - 1$ such that $v_i$ is a true twin of $v_{n}$, but then the triple $\{i, n-1, n\}$ contradicts the definition of interval orderings.
  Therefore,  $N[v_{n}] = K\cup \{v_n\}$, and replacing $K$ by $N[v_{n}]$ gives the desired clique path for $G$.

  Now that $v_{n}$ is the only simplicial vertex in $N[v_{n}]$, every maximal clique of $G - \{v_n\}$ is a maximal clique of $G$.
  By the definition of interval orderings, every vertex in $N(v_{n})$ is adjacent to $v_{n-1}$.  Thus, adding $N[v_{n}]$ after $K$ gives the desired clique path for $G$.
\end{proof}

Theorem~\ref{lem:end-vertex} has the following corollary, which has been observed by Gimbel~\cite{gimbel-88-end-vertices}.  What concerned Gimbel are the end intervals; an \emph{end interval} is an interval with the smallest left endpoint or the largest right endpoint in an interval representation.
Although by the end vertex theorem, every end vertex can be represented by an end interval, the other direction does not hold in general;
e.g., neither of vertices $2$ and $3$ in Figure~\ref{fig:bull} is an end vertex but tey can have end intervals.  Built on
Lekkerkerker and Boland~\cite{lekkerkerker-62-interval-graphs}, Gimbel gave a complete characterization of end intervals, of which Lemma~\ref{lem:bull} is one part.
\begin{lemma}[\cite{gimbel-88-end-vertices}]
  \label{lem:bull}
  Let $G$ be an interval graph.  An end vertex of $G$ cannot be the degree-2 vertex of a bull (vertex $5$ in Figure~\ref{fig:bull}).
\end{lemma}

\begin{figure}[h]
  \centering
  \begin{tikzpicture}[every node/.style={filled vertex}, scale=1.5]
    \foreach \i in {1, ..., 4} {
      \node ["$\i$" below] (v\i) at (\i, 0) {};
    }
    \node["$5$"] (u) at (2.5, .7) {};
    \draw (v1)-- (v2)-- (v3)-- (v4) (v2)-- (u)--(v3);
  \end{tikzpicture}
  \caption{Bull.}
  \label{fig:bull}
\end{figure}

By Theorem~\ref{lem:end-vertex}, an \textsc{lbfs} of the bull, disregarding where it starts, ends with a degree-1 vertex.  Let us use the bull to motivate our first variation of \textsc{lbfs}.   Intuitively, to use \textsc{lbfs} to recognize an interval graph, it makes more sense to start from one end vertex, i.e., $1$ or $4$ in the bull.  However, an \textsc{lbfs} starting from vertex $1$ cannot distinguish vertices $3$ and $5$, while an \textsc{lbfs} starting from vertex $4$ can.
If we conduct \textsc{lbfs} twice, from $1$ and $4$ respectively, then they may produce complementary information on the graph, which can be used to get a full picture of the graph.
Formally, for an \textsc{lbfs} ordering $\sigma$ of $G$ and $v\in V(G)$, we use $N_{\sigma}(v)$ to denote the set of neighbors of $v$ that are earlier than $v$ in $\sigma$, i.e.,
\[
  N_{\sigma}(v) = \{u\in N(v)\mid u<_{\sigma} v\}.
\]
The \emph{$v$-snapshot} in $\sigma$, denoted by $S_{\sigma}(v)$, comprises the vertices that are adjacent to all vertices in $N_{\sigma}(v)$ and that are visited not earlier than $v$.  Note that $S_{\sigma}(v)$ is precisely $S$ defined in step~2.1 in the $(\sigma(v))$th iteration.  There are precisely $n$ snapshots in an \textsc{lbfs}, of which the first is always $V(G)$ itself.  Other snapshots, which are different from $V(G)$, are called \emph{proper snapshots}.  The label of every vertex in $S_{\sigma}(v)$ is $N_{\sigma}(v)$ at the moment $v$ is visited, and thus we use $N_{\sigma}(S_{\sigma}(v))$ interchangeably with $N_{\sigma}(v)$.
For an \textsc{lbfs} of the bull starting from vertex $1$, the third snapshot is $\{3, 5\}$.  Not able  to distinguish them from this side, we check the other direction---in any \textsc{lbfs} starting from vertex $4$, vertex $3$ is visited before vertex $5$---and thus we probably should visit vertex $5$ first.

Instead of running \textsc{lbfs} twice from different ends, Simon~\cite{simon-91-interval} proposed a new procedure.  Apart from the graph $G$, the procedure \textsc{lbfs}$^{+}(G, \sigma)$ takes an \textsc{lbfs} ordering of $G$ as input, and it replaces step~2.2 of procedure \textsc{lbfs} (Figure~\ref{fig:alg-lbfs}) by

\begin{itemize}
\item[2.2.] $\quad v\leftarrow$ the last vertex of $\sigma|_{S}$;
\end{itemize}
We usually use $\sigma^{+}$ to denote the output of $\textsc{lbfs}^{+}(G, \sigma)$.  Unlike the \textsc{lbfs} itself, which may output different ordering dependent on the vertex selection in step~2.2, an \textsc{lbfs}$^{+}$ ordering is unique for any given \textsc{lbfs} ordering $\sigma$.  For example, if $\sigma = \langle 1, 2, \ldots, 8\rangle$ for the graph in Figure~\ref{fig:bfs-example}, then $\sigma^{+}$ has to be $\langle 8, 2, 7, 6, 5, 4, 3, 1\rangle$.  Indeed, for a complete graph, any ordering $\sigma$ is a valid \textsc{lbfs} ordering, while $\sigma^{+}$ has to be the full reversal of $\sigma$ because every snapshot of $\sigma^{+}$ consists of all the unvisited vertices.

For any graph $G$ and any \textsc{lbfs} ordering $\sigma$ of $G$, the first snapshot of $\sigma^{+}$ is $V(G)$, and thus its first vertex is the last vertex of $ \sigma|_{V(G)}=\sigma$, i.e., the end vertex of $\sigma$.  
Corneil et al.~\cite{corneil-09-lbfs-strucuture-and-interval-recognition, corneil-10-end-vertices-lbfs} characterized end vertices of \textsc{lbfs}$^{+}$ orderings of an interval graph.

\begin{lemma}[The flipping lemma \cite{corneil-09-lbfs-strucuture-and-interval-recognition, corneil-10-end-vertices-lbfs}]
  \label{lem:flipping}
  Let $G$ be an interval graph.  A vertex $z$ is an end vertex if and only if for any \textsc{lbfs} ordering $\sigma$ of $G$ with $\sigma(z) = 1$, the ordering $\textsc{lbfs}^{+}(G, \sigma)$ ends with $z$.
\end{lemma}

One may check that \textsc{lbfs}$^{+}$ works perfectly for the bull (which is a unit interval graph; see Section~\ref{sec:uig}).  However, it fails fatally for the graph in Figure~\ref{fig:bfs-example}: Consider the \textsc{lbfs} orderings $\sigma = \langle 1, 2, 5, 4, 6, 3, 7, 8\rangle$ for the graph in Figure~\ref{fig:bfs-example},  and then $\sigma^{+} = \langle 8, 2, 7, 6, 4, 5, 3, 1\rangle$, neither of which can distinguish vertices $4$ and $5$ correctly.

A subset $M$ of vertices forms a \emph{module} of $G$ if for any pair of vertices $u,v \in M$, a vertex $x \not\in M$ is adjacent to $u$ if and only if it is adjacent to $v$ as well; e.g., $\{3, 4, \ldots, 7\}$ of the graph in Figure~\ref{fig:bfs-example}.
 The set $V(G)$ and all singleton vertex sets are modules, called \emph{trivial}.   A graph on four or more vertices is \emph{prime} if it contains only trivial modules.
  The following observation of Hsu~\cite{hsu-95-recognition-cag} is behind Hsu and Ma's recognition algorithms for interval graphs~\cite{hsu-99-recognizing-interval-graphs}.

\begin{theorem}[\cite{hsu-95-recognition-cag}]
  \label{thm:prime}
  A prime interval graph has a unique clique path, up to full reversal.
\end{theorem}

Two adjacent vertices that form a module are called \emph{true twins}; note that $u$ and $v$ are twins if and only if $N[u] = N[v]$.
One can slightly strengthen Theorem~\ref{thm:prime} by weakening its condition to allow true twins.
\begin{corollary}\label{cor:prime}
  Let $G$ be an interval graph.  If every nontrivial module of $G$ is a set of true twins, then $G$ has a unique clique path, up to full reversal.
\end{corollary}
\begin{proof}
  Let $u$ and $v$ be true twins.  Then a set $K$ of vertices is a maximal clique of $G$ if and only if $K\setminus \{u\}$ is a maximal clique of $G - \{u\}$.  Moreover, $\langle K_1, K_2, \ldots, K_{\ell}\rangle$ is a clique path of $G$ if and only if $\left\langle K_1\setminus \{u\}, K_2\setminus \{u\}, \ldots, K_{\ell}\setminus \{u\}\right\rangle$ is a clique path of $G-\{u\}$.
\end{proof}

\section{An appetizer: unit interval graphs}\label{sec:uig}

In a \emph{unit interval representation}, every interval has the same length, and a graph is a \emph{unit interval graph} if it has a unit interval representation.   For example, the bull is a unit interval graph, while the graph in Figure~\ref{fig:bfs-example} is not.
An interval representation is \emph{proper} if no interval in the representation properly contains another, and \emph{proper interval graphs} are defined accordingly.
A unit interval representation is  necessarily proper, but the other way does not hold true in general.  A nontrivial observation of Roberts~\cite{roberts-69-indifference-graphs} states that
these two subclasses of interval graphs actually coincide.
This section gives a self-contained presentation of a linear-time algorithm that uses \textsc{lbfs}  and \textsc{lbfs}$^{+}$ to recognize proper interval graphs.  Every statement toward the main result will be derived from scratch, though not necessarily by formal proofs.  (Since we are not proving the equivalence between unit interval graphs and proper interval graphs, strictly speaking, it is not self-contained for recognizing unit interval graphs.)

An ordering $\sigma$ of $V(G)$ is an \emph{umbrella ordering} if for any triple of vertices $u, v, w$ of $G$ with $u <_{\sigma} v <_{\sigma} w$, vertices $u$ and $w$ are adjacent if and only if they are both adjacent to $v$.
Given a proper interval representation, the left endpoints of all vertices, from the smallest to the largest, with ties broken arbitrarily, induce an ordering of $V(G)$.  One can obtain the same ordering by considering all the right endpoints.
On the one hand, it is trivial to see that this ordering is an umbrella ordering.  On the other hand, from an umbrella ordering $\sigma$ of a graph $G$, we can construct a proper interval representation by setting
\begin{equation}
  \label{eq:ui-model}
  \tag{PI}
  I(v) = \left[ \sigma(v), \sigma(u) + \frac{\sigma(v)}{n} \right],  
\end{equation}
where $u$ is the last vertex in $\sigma|_{N[v]}$.  We rely on the reader to verify that the resulting representation is indeed proper.  Therefore, a graph $G$ is a unit interval graph if and only if it has an umbrella ordering \cite{looges-93-greedy-algorithms-uig}.
From \eqref{eq:ui-model} one can also see that a unit interval graph has a proper representation in which all the endpoints are integers between $1$ and $n^2$.  On the other hand, in general, a unit representation with only integral endpoints has to use very large integers, and this suggests the difficulty of building unit representations.
The following fact is immediate from the definition of umbrella orderings, and we can also see it through the representation $\mathcal{I}$ derived with \eqref{eq:ui-model} from an umbrella ordering: The reversal of $\sigma$ is precisely the right endpoints of the intervals in $\mathcal{I}$ in decreasing order.
\begin{proposition}[Folklore]\label{prop:ui-ordering-reversal}
  Let $G$ be a unit interval graph.  An ordering $\sigma$ of $V(G)$ is an umbrella ordering of $G$ if and only if the reversal of $\sigma$ is an umbrella ordering of $G$.
\end{proposition}

We have mentioned in Section~\ref{sec:lbfs} how to verify a given ordering $\sigma$ of $V(G)$ is an interval ordering of the graph $G$.  Conducting this test twice, once for $\sigma$, and the other for the reversal of $\sigma$, would verify whether $\sigma$ is an umbrella ordering: The answer is ``yes'' if and only it both $\sigma$ and $\sigma^{+}$ pass the test.  Indeed, the test on $\sigma$ verifies $u w\in E(G)\Rightarrow u v \in E(G)$, and the other verifies $u w\in E(G)\Rightarrow v w\in E(G)$, both for all $u <_{\sigma} v <_{\sigma} w$.\footnote{Corneil et al.~\cite{corneil-95-recognition-unit-interval} presented another way to verify umbrella orderings, which needs to collect information from the \textsc{lbfs} procedure.  The one given here is divorced from the construction of the ordering, hence conceptually simpler, though it is inferior in terms of performance.
  Another benefit is that our verification procedure better reveals the connection between umbrella orderings and interval orderings.}
In the rest we will be focused on unit interval graphs.

The claw in Figure~\ref{fig:claw} is an interval graph but not a unit interval graph.  We rely on the reader to check that there cannot be a way of arranging a proper interval representation for the four vertices in the claw.
The non-existence of claws forces unit interval graphs to have very simple clique paths.  Recall that for a vertex $v$, we use $\lp{v}$ and $\rp{v}$ to denote the indices of the first and, respectively, last maximal cliques containing $v$ in a clique path.

\begin{figure}[h]
  \centering\small
  \begin{tikzpicture}[every node/.style={filled vertex}, scale=1]
      \node (v0) at (2, 2) {};
      \foreach[count=\i] \j in {1, 5, 8} {
        \node (v\j) at (\i, 1) {};
        \draw (v\j)-- (v0);
      }
    \end{tikzpicture}
  \caption{Claw}
  \label{fig:claw}
\end{figure}

\begin{proposition}\label{prop:claw-free}
  Let $G$ be an unit interval graph, and let $\langle K_1, K_2, \ldots, K_{\ell}\rangle$ be a clique path of $G$.
  \begin{enumerate}[(i)]
  \item If $K_p$, $1\le p\le \ell$, contains a simplicial vertex, then for every $v\in K_p$, at least one of $\lp{v}$ and $\rp{v}$ is $p$.
  \item If $G$ is connected, then for any $p = 2, \ldots, \ell$, there is a vertex $v$ with $\lp{v} < \rp{v} = p$.
  \end{enumerate}
\end{proposition}
\begin{proof}
  Suppose for contradiction to assertion (i) that $\lp{v} < p$ and $\rp{v} > p$.  Since both $K_{p-1}$ and $K_{p+1}$ are maximal cliques of $G$, there must be a vertex $x\in K_{p-1}\setminus K_{p}$, and a vertex $z\in K_{p+1}\setminus K_{p}$.  By assumption, there is a simplicial vertex $y$ with $\lp{y} = \rp{y} = p$.  But then $\{v, x, y, z\}$ induces a claw, which is impossible.

  For assertion (ii), since $K_{p}$ is a maximal clique, there exists $x$ with $\rp{x} = p$.  If $\lp{x} < p$, then we are done.  Now that $\lp{x} = p$, then $x$ is simplicial.  By assertion (i), every vertex in $K_{p-1}\cap K_{p}$, which is nonempty because $G$ is connected, is disjoint from $K_{p+1}$.  Therefore, there always exists some vertex $v$ with $\lp{v} < \rp{v} = p$.
\end{proof}

Let $\langle K_1, K_2, \ldots, K_{\ell}\rangle$ be a clique path of a connected unit interval graph $G$.
From Proposition~\ref{prop:claw-free}(i)
we can conclude that any simplicial vertex $v$ in $K_i$ with $1< i < \ell$ is the degree-two vertex of a bull.  By Lemma~\ref{lem:bull}, $v$ cannot be end vertex, and by the perfect elimination theorem (Theorem~\ref{thm:peo}), all the end vertices of $G$ are in $K_{1}$ and $K_{\ell}$.  Therefore, by Theorem~\ref{thm:ordering-path}, the two ends of any clique path of $G$ must be $K_{1}$ and $K_{\ell}$.  The same argument applies to the connected unit interval subgraph induced by $\bigcup_{i=2}^{\ell - 1}$, with clique path $\langle K_2, \ldots, K_{\ell-1}\rangle$.  It is nontrivial but one can show that $\langle K_{\ell}, K_2, K_3, \ldots, K_{\ell-1}, K_{1}\rangle$ is not a clique path of $G$.  We can continue this argument to conclude that a connected unit interval graph has a unique clique path.  Yet another way to derive this fact is through Corollary~\ref{cor:prime} and the simple structure of a unit interval graph that has a universal vertex.  If a connected unit interval graph $G$ contains a non-clique module $M$, then $N(M)$ comprises universal vertices, which has no impact on the arrangement of clique paths, and $G[M]$ either is a connected unit interval graph in which every nontrivial module is a set of true twins, or consists of two disjoint cliques with no edges in between.  For the sake of completeness, we give a direct and simple proof.
\begin{theorem}[\cite{panda-09-bicompatible-elimination-ordering}]
  \label{thm:unique-clique-path}
  A connected unit interval graph has a unique clique path, up to full reversal.
\end{theorem}
\begin{proof}
  Let $G$ be a connected unit interval graph with $\ell$ maximal cliques, and let $\langle K_1, K_2, \ldots, K_{\ell}\rangle$ be a clique path of $G$, denoted by $\mathcal{K}$.  Suppose to the contradiction of the theorem that there is another clique path $\mathcal{K}'$ of $G$ that is neither $\mathcal{K}$ nor its reversal.  We can find a minimal subsequence $\langle K_{p}, \ldots, K_{q}\rangle$, $1\le p \le q \le \ell$ such that they appear neither in this order or its reversal in $\mathcal{K}'$.
  Let $G'$ be the subgraphs induced by $\bigcup_{i=p}^{q}K_{i}$, and let $\mathcal{K}''$ denote the sequence of these cliques as they appear in $\mathcal{K}'$.
  It is easy to use the definition of clique paths to verify that both $\langle K_{p}, \ldots, K_{q}\rangle$ and $\mathcal{K}''$ are clique paths of $G'$. 
  By the minimality, one of $K_{p}$ and $K_{q}$ is at the end of $\mathcal{K}''$.  We may assume without loss of generality that $K_{p}$ is at one end of $\mathcal{K}''$, then by the minimality, $\mathcal{K}''$ has to be
\[
  K_{p}, K_{p+1}, \ldots, K_{r-1}, K_{q}, K_{q-1}, \ldots, K_{r+1}, K_{r}
  \]
  for some $r$ with $p < r< q$.  Since $G$ is connected, there is a vertex $v$ in $K_{r-1}\cap K_{q}$, and by the definition of clique paths, $v$ is also in $K_{r}$.  On the other hand, since $K_{r}$ is the last clique in $\mathcal{K}''$, it contains a simplicial vertex of $G'$.  We have thus a contradiction to proposition~\ref{prop:claw-free}(i).
\end{proof}

It should not be surprising that we can transform proper interval representations, clique paths, and umbrella orderings of a unit interval graph to each other. 
Since there is no special requirement of clique paths of a unit interval graph, transforming proper interval representation to a clique path is the same as a general interval representation, while transforming an umbrella ordering $\langle v_1, \ldots, v_{n}\rangle$ to a clique path can be done as follows.  For each vertex $v_i$ with the last neighbor $v_j$, the vertex set $\{v_{i}, \ldots, v_{j}\}$ is a clique, and keeping the maximal ones in the original order gives a clique path.
On the other hand, the transformations from a clique path to the other two are more subtle.  
Let $G$ be a connected unit interval graph.  We may assume without loss of generality that $G$ is not complete, and let $\langle K_1, K_2, \ldots, K_{\ell}\rangle$ be a clique path of $G$.
To derive a proper interval model, we need to set $I(v) = [\lp{v} - x, \rp{v} + y]$, where $x$ and $y$, which are required to ensure that the model is proper, can be calculated in a similar manner as the fractional number in \eqref{eq:ui-model}.
In an umbrella ordering, a pair of true twins can appear in an arbitrary order.
For a pair of vertices $u$ and $v$ that are not true twins,
 $u<_{\sigma} v$ if and only if 

 \begin{align}
  \lp{u} &< \lp{v} \text{ or } \notag
  \\
  \lp{u} &= \lp{v} \text{ and } \rp{u} < \rp{v}.
  \tag{UO} \label{eq:ui-ordering}
\end{align}
Theorem~\ref{thm:unique-clique-path} has the following corollary.

\begin{corollary}\label{cor:ui-ordering}
  Let $G$ be a connected unit interval graph, and let $\sigma_1$ and $\sigma_2$ be umbrella orderings of $G$.  Then $\sigma_2$ can be obtained from $\sigma_1$ or its reversal by re-ordering true twins.
\end{corollary}

Corollary~\ref{cor:ui-ordering} implies the famous characterization of Deng et al.~\cite{deng-96-proper-interval-and-cag}, which states that a unit interval graph without true twins has a unique umbrella ordering, up to full reversal.
To find an umbrella ordering of $G$, it suffices to use characterization \eqref{eq:ui-ordering} to decide the order of each pair of vertices.
It is not difficult to see that if an \textsc{lbfs} starts from a simplicial vertex in $K_1$, then it visits the maximal cliques of $G$ on the clique path one by one, and hence is able to tell whether $ \lp{u} < \lp{v}$.  On the other hand, an \textsc{lbfs} from a simplicial vertex in $K_{\ell}$ is able to tell whether $ \rp{u} < \rp{v}$.  Combining them we are able to recognize unit interval graphs.
Recall that an end vertex is the last vertex of some \textsc{lbfs} ordering of $G$, and note that the first assertion of the following theorem is the end vertex theorem (Theorem~\ref{lem:end-vertex}), and the proof given here works only for unit interval graphs.

\begin{lemma}\label{lem:uig-snapshot}
  Let $G$ be a connected unit interval graph with $\ell$ maximal cliques.  Let $\langle K_1, K_2, \ldots, K_{\ell}\rangle$ be a clique path, and $\sigma$ an \textsc{lbfs} ordering of $G$.
  \begin{enumerate}[(i)]
  \item    The last vertex of $\sigma$ is a simplicial vertex in $K_{1}$ or $K_{\ell}$.
  \item If $\sigma$ starts from an end vertex in $K_1$, then for $i=1, \ldots, \ell-1$, vertices in $K_{i}$ are visited before those in $K_{i+1} \setminus K_{i}$.
  \item If $\sigma$ starts from an end vertex in $K_1$, then for each proper snapshot $S\not\subseteq K_1$, there is $p\in \{2, \ldots, \ell\}$ such that $S\subseteq K_{p}\setminus K_{p-1}$.
  \end{enumerate}
\end{lemma}
\begin{proof}
  All the assertions hold trivially or vacuously when $G$ is complete.  Henceforth we assume that $G$ is not complete, hence $\ell > 1$.
  Let $v$ be the first vertex of $\sigma$.
  We first show that for every pair of vertices $x$ and $y$, if $\rp{v} < \lp{x} < \lp{y}$, then $x<_{\sigma} y$.
  Suppose for contradiction that $\rp{v} < \lp{x} < \lp{y}$ but $y<_{\sigma} x$, and let $x,y$ be chosen in way that (A1) $\lp{x}$ is the smallest among all such pairs; and (A2) $y$ is the first in $\sigma$ for the fixed $x$.  By assumption (A2), $\lp{u} \le \lp{x}$ for every vertex $u\in N_{\sigma}(y)\setminus N(v)$.  Thus, if a vertex in $N_{\sigma}(y)$ is adjacent to $y$, then by the definition of clique paths, it has to be adjacent to $x$ as well.
  By Proposition~\ref{prop:claw-free}(ii), there exists some vertex $w$ with $\lp{w} < \rp{w} = \lp{x}$, then $w\in N(x) \setminus N(y)$.  Note that $w<_{\sigma} y$: It follows from the procedure \textsc{lbfs} if $w$ is adjacent to $v$, and by assumption (A1) otherwise.  In summary, the label of $y$ is a proper subset of that of $x$ when it is visited, which is impossible.
  By Proposition~\ref{prop:ui-ordering-reversal} and a symmetric argument as above, we can conclude that for every pair of vertices $x$ and $y$, if $\rp{y} < \rp{x} < \lp{v}$, then $x<_{\sigma} y$.

  (i) We may assume that $v$ is not a universal vertex.
  Otherwise we may produce another ordering $\sigma'$ from $\sigma$ by exchanging the first non-universal vertex in $\sigma$ and $v$.
  It is easy to verify that $\sigma$ is an \textsc{lbfs} ordering of $G$ if and only if $\sigma'$ is.  Note that the last vertex of $\sigma'$ is the last vertex of $\sigma$.
  By procedure \textsc{lbfs}, the last vertex $z$ of $\sigma$ is not adjacent to $v$.  Suppose without loss of generality that $\lp{z} > \rp{v}$.  Then by the argument above, we must have $\lp{z} = \ell$; in other words, $z$ is a simplicial vertex in $K_{\ell}$.  A symmetric argument concludes that $z$ is a simplicial vertex in $K_{1}$ if $\rp{z} < \lp{v}$.

  (ii) For $i = 1$, it is because all vertices in $K_1$ are adjacent to the first vertex in $\sigma$ while vertices in $K_2\setminus K_1$ are not.  For $i \ge 2$, the statement follows from the argument above.

  (iii) Suppose that $w$ is the first vertex in $\sigma|_{S}$, i.e., $S=S_{\sigma}(w)$, and we show that $p = \lp{w}$ is the required index.
  By Proposition~\ref{prop:claw-free}(ii), there exists a vertex $u$ with $\lp{u} < \rp{u} = \lp{w}$.  By (ii),  $u<_{\sigma}w$, and $S$ is disjoint from $K_1, \ldots, K_{p-1}$.  Thus, $S\subseteq N(u)\setminus \bigcup_{i=1}^{p-1}K_{i}\subseteq K_{p}\setminus K_{p-1}$.
\end{proof}

We are now ready to prove the main theorem of this section.
The key observation is that any pair of vertices that are not true twins can be distinguished by an \textsc{lbfs} from one end, while an \textsc{lbfs}$^{+}$ ordering combines information from both.  For example, vertices $3$ and $5$ in Figure~\ref{fig:bull} cannot be distinguished by an \textsc{lbfs} from vertex $1$, but can be distinguished by any \textsc{lbfs} from vertex $4$.  It is the other way for vertices $2$ and $5$.

\begin{theorem}\label{thm:uig}
    Let $G$ be a unit interval graph, and let $\sigma$ be an \textsc{lbfs} ordering of $G$.  If $\sigma$ starts from an end vertex of $G$, then $\textsc{lbfs}^{+}(G, \sigma)$ is an umbrella ordering of $G$.
\end{theorem}
\begin{proof}
  Let $\sigma^{+} = \textsc{lbfs}^{+}(G, \sigma)$.
  We may renumber the vertices such that $\sigma^{+}(v_i) = i$.  Then by the procedure \textsc{lbfs}$^{+}$, $\sigma(v_1)  = n$.
  It suffices to show that $p < q$ for every pair of vertices $v_p$ and $v_q$ in $G$ satisfying \eqref{eq:ui-ordering}.
  If $\lp{u} < \lp{v}$, then $p < q$ follows from Lemma~\ref{lem:uig-snapshot}(ii).  Now that $\lp{u} = \lp{v}$ and $\rp{u} < \rp{v}$; note that $v_p$ and $v_q$ are adjacent in $G$.
  When $\sigma^{+}$ visits the first of $v_p$ and $v_q$, the other is also in the snapshot.  By Lemma~\ref{lem:uig-snapshot}(ii), applied to the reversal of the clique path, $v_q <_{\sigma} v_p$, and thus  $\sigma^{+}$ should choose $v_p$.  Therefore, we always have $v_p <_{\sigma^{+}} v_q$, and this concludes the proof.
\end{proof}

Described in Figure~\ref{fig:alg-uig-3}  is the algorithm from Corneil~\cite{corneil-04-recognize-uig}.  By procedure \textsc{lbfs}$^{+}$, the first vertex of $\sigma$ is an end vertex.  If $G$ is a unit interval graph, then by Theorem~\ref{thm:uig}, $\sigma^{+}$ is an interval ordering of $G$. On the other hand, any ordering is incorrect if $G$ is not a unit interval graph.

\begin{figure}[h!]
  \centering 
  \begin{tikzpicture}
    \path (0,0) node[text width=.65\textwidth, inner xsep=20pt, inner ysep=10pt] (a) {
      \begin{minipage}[t!]{\textwidth}
        \begin{tabbing}
          Output: \= \kill
          Input: \> A connected graph $G$.
          \\
          Output: Whether $G$ is a unit interval graph.
        \end{tabbing}        

        \begin{tabbing}
          AAA\=Aaa\=aaa\=Aaa\=MMMMMMAAAAAAAAAAAAA\=A \kill
          1.\> $\tau\leftarrow$ an \textsc{lbfs} ordering of $G$;
          \\
          2.\> $\sigma\leftarrow \textsc{lbfs}^{+}(G, \tau)$;
          \\
          3.\> $\sigma^{+}\leftarrow \textsc{lbfs}^{+}(G, \sigma)$;
          \\
          4.\> \textbf{if} $\sigma^{+}$ is an umbrella ordering of $G$ \textbf{then return} ``yes'';
          \\
          5.\> \textbf{else return} ``no.''
        \end{tabbing}
      \end{minipage}
    };
    \draw[draw=gray!60] (a.north west) -- (a.north east) (a.south west) -- (a.south east);
  \end{tikzpicture}
  \caption{The three-sweep recognition algorithm for unit interval graphs~\cite{corneil-04-recognize-uig}.}
  \label{fig:alg-uig-3}
\end{figure}

Some remarks on Theorem~\ref{thm:uig}.  Our statement is slightly more general than the one made by Corneil~\cite{corneil-04-recognize-uig}; in particular, we only require $\sigma$ to start from an end vertex, and it does not need to be an \textsc{lbfs}$^{+}$ ordering.  As a result, $\tau$ does not need to be an \textsc{lbfs} ordering either.  Since the only purpose of the first sweep is to find an end vertex, it can be replaced by \textsc{bfs}.  The following lemma from an earlier recognition algorithm of Corneil et al.~\cite{corneil-95-recognition-unit-interval} can help us to find an end vertex with \textsc{bfs}.  Indeed, they developed a recognition algorithm for unit interval graphs using only \textsc{bfs}.  Note that this statement does not apply to interval graphs.
\begin{lemma}[\cite{corneil-95-recognition-unit-interval}]
  Let $G$ be unit interval graph.  Let $T$ be a \textsc{bfs} tree of $G$.  A vertex in the last level with the minimum degree is an end vertex.
\end{lemma}

Recall that setting $I(v) = [\lp{v}, \rp{v}]$ for every vertex gives an interval representation.  Lemma~\ref{lem:uig-snapshot}(ii) has the following implication.  

\begin{corollary}
  Let $G$ be a unit interval graph.  If we start \textsc{lbfs} from an end vertex of $G$, then the result is an interval ordering of $G$.
\end{corollary}

For the purpose of making an interval ordering, an \textsc{lbfs} from a simplicial vertex in $K_1$ does not need to distinguish vertices in $K_{p}\setminus K_{p-1}$ for $p = 1, \ldots, \ell$.  In other words, we only need to consider the first row of \eqref{eq:ui-ordering}.
For two vertices $u$ and $v$ with $\lp{u} = \lp{v}$, we have $\rp{u} < \rp{v}$ if and only if the degree of $u$ is strictly smaller than $v$.  Therefore, to distinguish such a pair of vertices, we do not really do two sweeps, and it suffices to use the vertex degrees.
We are thus motivated to define another variation of \textsc{lbfs}, which always chooses a vertex of the minimum degree from the current proper snapshot.
Apart from the graph $G$, the procedure \textsc{lbfs}$^{\delta}(G, u)$ takes an end vertex $u$ of $G$ as input.   It replaces step~2.2 of procedure \textsc{lbfs} (Figure~\ref{fig:alg-lbfs}) by
\begin{itemize}
\item[2.2.] \quad \textbf{if} $i=1$ \textbf{then} $v\leftarrow u$;
  \vspace*{-8pt}
\item[] \quad\textbf{else} $v\leftarrow$ a vertex with the minimum degree in $S$;
\end{itemize}
We remark that a variation of \textsc{lbfs} that chooses a largest-degree vertex has been used by Hsu and Ma~\cite{hsu-99-recognizing-interval-graphs} for finding modules of a chordal graph.  We have thus a two-sweep recognition algorithm for unit interval graphs, as described in Figure~\ref{fig:alg-uig-2}.

\begin{figure}[h!]
  \centering 
  \begin{tikzpicture}
    \path (0,0) node[text width=.65\textwidth, inner xsep=20pt, inner ysep=10pt] (a) {
      \begin{minipage}[t!]{\textwidth}
        \begin{tabbing}
          Output: \= \kill
          Input: \> A connected graph $G$.
          \\
          Output: Whether $G$ is a unit interval graph.
        \end{tabbing}        

        \begin{tabbing}
          AAA\=Aaa\=aaa\=Aaa\=MMMMMMAAAAAAAAAAAAA\=A \kill
          1.\> $u\leftarrow$ an end vertex of $G$;
          \\
          2.\> $\sigma\leftarrow \textsc{lbfs}^{\delta}(G, u)$;
          \\
          4.\> \textbf{if} $\sigma$ is an umbrella ordering of $G$ \textbf{then return} ``yes'';
          \\
          5.\> \textbf{else return} ``no.''
        \end{tabbing}
      \end{minipage}
    };
    \draw[draw=gray!60] (a.north west) -- (a.north east) (a.south west) -- (a.south east);
  \end{tikzpicture}
  \caption{A two-sweep recognition algorithm for unit interval graphs.}
  \label{fig:alg-uig-2}
\end{figure}

\begin{theorem}
    Let $G$ be a unit interval graph.  If $u$ is an end vertex of $G$, then $\textsc{lbfs}^{\delta}(G, u)$ is an umbrella ordering of $G$.
\end{theorem}

Tarjan and Yannakakis~\cite{tarjan-84-chordal-recognition} proposed another graph search algorithm for the purpose of recognizing chordal graphs, among others.
Also based on adjacencies with visited vertices, an
\emph{maximum cardinality search} (\textsc{mcs}) chooses an unvisited vertex that has the maximum number of visited neighbors.  Lemma~\ref{lem:uig-snapshot}(ii) can be re-interpreted as: An \textsc{lbfs} of a unit interval graph from an end vertex visits the maximal cliques one by one.  Therefore, it is also an \textsc{mcs} \cite{cao-19-end-vertices}, and we have the following corollary, where \textsc{mcs}$^{\delta}$ is defined in a similar spirit as \textsc{lbfs}$^{\delta}$.

\begin{corollary}
  Let $G$ be a unit interval graph.  If $u$ is an end vertex of $G$, then $\textsc{mcs}^{\delta}(G, u)$ is an umbrella ordering of $G$.
\end{corollary}

\section{Interval graphs}\label{sec:algorithm}

Compared to unit interval graphs, an \textsc{lbfs} of an interval graph can jump in an unpredictable and sometimes arbitrary way because of the existence of ``long intervals.''
The graph in Figure~\ref{fig:the-graph} was devised by Corneil et al.~\cite{corneil-09-lbfs-strucuture-and-interval-recognition}.
Since it is very handy for our explanation, we will use it as the main example for this section, and we henceforth referred to it as $G^{\star}$.
As indicated by the integral points in the figure, $G^{\star}$ has $16$ maximal clique.  Since $G^{\star}$ is prime, the arrangement of maximal cliques is unique by Theorem~\ref{thm:prime}, though the graph admits many different valid interval orderings.
After visiting vertices $1$ and $2$ in $G^{\star}$, an \textsc{lbfs} may visit any vertex between $3$ and $20$ as the third vertex.
In general, an \textsc{lbfs} ordering of an interval graph can start from an end vertex, and then quickly jump to another vertex in an arbitrary position in the interval representation, even when the graph is prime.
Therefore, one should be extremely careful when talking about the ``left'' and the ``right'' for an interval graph $G$.
  
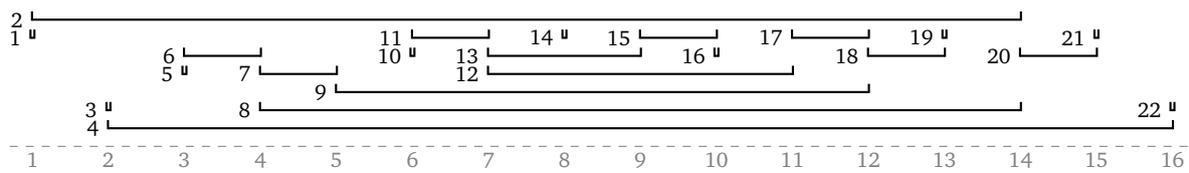
\begin{figure}[h]
  \centering\scriptsize
    \begin{tikzpicture}[yscale=1.2]
      \begin{scope}[every path/.style={{|[right]}-{|[left]}}]
        \foreach \x in {1, ..., 16} \node[gray] at (\x-1, -.15) {$\x$};
        \foreach[count=\i] \l/\r/\y/\c in {-.02/.02/6/, 0/13/7/,
          0.98/1.02/2/, 1/15/1/, 
          1.98/2.02/4/, 2/3/5/, 3/4/4/,
          3/13/2/, 4/11/3/, 4.98/5.02/5/,
           5/6/6/,
          6/10/4/, 6/8/5/, 6.98/7.02/6/,
          8/9/6/, 8.98/9.02/5/, 10/11/6/, 11/12/5/, 11.98/12.02/6/, 13/14/5/, 
          13.98/14.02/6/, 14.98/15.02/2/}{
          \draw[\c, thick] (\l-.02, \y/5) node[left] {$\i$} to (\r+.02, \y/5);
        }
      \end{scope}
      \draw[dashed, gray] (-.3, 0) -- (15.3, 0);
    \end{tikzpicture}
    \caption{An interval graph $G^{\star}$, presented as an interval representation.}
  \label{fig:the-graph}
\end{figure}

Listed in Figure~\ref{fig:example-orderings} are six \textsc{lbfs} orderings of $G^{\star}$.  In particular, for $i = 1, 2, 3$, $\sigma_i^{+} = \textsc{lbfs}^{+}(G^{\star}, \sigma_i)$.  The reader unfamiliar with $\textsc{lbfs}^{+}$ is suggested to go through these orderings before proceeding.
\begin{figure}[h]
  \centering\small
  \begin{align*}
    \sigma_1&: 1, 2, 20, 8, 4, [19, 18, 17, 9, 12, 16, 15, 13, 11, 14, 10, 7, 6], 5, 3, 21, 22.
    \\
    \sigma_1^{+}&: 22, 4, 21, 20, 8, 2, [6, 7, 9, 10, 11, 13, 12, 14, 15, 16, 17, 18, 19], 5, 3, 1.
    \\[3mm]
    \sigma_2&: 1, 2, 3, 4, 8, [20, [15, 16, 12, 9, 13, 11, 14, 17, 10, 18, 7, 19, 6]], 5, 21, 22.
    \\
    \sigma_2^{+}&: 22, 4, 21, 20, 8, 2, [6, 7, 9, 18, 17, 12, 14, 13, 11, 15, 16, 10, 19], 5, 3, 1.
    \\[3mm]
    \sigma_3&: 1, 2, [4, 20, 8, [6, 7, 9, 18, 17, 12, [11, 13, 15, 14, 16], 10, 19], 5, 3], 21, 22.
    \\
    \sigma_3^{+}&: 22, 4, 21, 20, 8, 2, [19, 18, 17, 9, 12, [16, 15, 13, 14, 11], 10, 7, 6], 5, 3, 1.
  \end{align*}
  \caption{For $i = 1, 2, 3$, $\sigma_i$ is an \textsc{lbfs} ordering of $G^{\star}$, and  $\sigma_i^{+} = \textsc{lbfs}^{+}(G^{\star}, \sigma_i)$.}
  \label{fig:example-orderings}
\end{figure}

Modules play a very similar role in interval graphs as true twins have played in unit interval graphs.
Recall that a set of true twins is a module.  Since an interval graph is chordal, if a module $M$ of an interval graph is not a clique, then $N(M)$ has to be a clique: Two nonadjacent vertices from $M$ and two nonadjacent vertices from $N(M)$ would induce  a $4$-cycle.
Let $M$ be a module of a graph $G$, and let $G'$ be the graph $G - (M\setminus\{v\})$ for some vertex $v\in M$.  Since we can always use the same intervals for true twins, if $M$ is a set of true twins, then $G$ is an interval graph if and only if $G'$ is.  In the general case, $M$ is not a clique.
We have mentioned that $N(M)$ is a clique; moreover, both $G'$ and $G[M]$ are induced subgraphs of $G$, hence also interval graphs.
It is known that these conditions are also sufficient \cite{hsu-99-recognizing-interval-graphs}.  Given any interval representation $\mathcal{I}_M$ for $G[M]$ and interval representation $\mathcal{I}'$ for $G'$, we can always project  $\mathcal{I}_M$ onto the interval for $s$ in $\mathcal{I}'$.  (Note that we can always modify an interval representation such that every interval has a positive length.)
If we take $G$ to be the graph in Figure~\ref{fig:bfs-example}, with $M = \{3, 4, \ldots, 7\}$ and $v = 5$, then the subgraph $G'$ and $G[M]$ are shown in Figure~\ref{fig:module-models}(a, b), and Figure~\ref{fig:module-models}(c) illustrates the projection of $\mathcal{I}_M$ onto $I(v)$ in $\mathcal{I}'$.

\begin{figure}[h]
  \centering\small
  \begin{subfigure}[b]{.19\linewidth}
    \centering
    \begin{tikzpicture}[every node/.style={filled vertex}, scale=1]
      \node ["$2$"] (v0) at (2, 2) {};
      \foreach[count=\i] \j in {1, 5, 8} {
        \node ["$\j$" below] (v\j) at (\i, 1) {};
        \draw (v\j)-- (v0);
      }
    \end{tikzpicture}
    \caption{}
  \end{subfigure}  
  \quad
  \begin{subfigure}[b]{.25\linewidth}
    \centering
    \begin{tikzpicture}[every node/.style={filled vertex}, scale=1]
      \foreach[count=\i from 2] \j in {3, 4, 6, 7} {
        \node ["$\j$" below] (v\j) at (\i, 1) {};
      }
      \node["$5$" below] (w) at (3.5, 0.7) {};
      \draw (v3)-- (v4)-- (v6)-- (v7) (v6)-- (w)--(v4);
    \end{tikzpicture}
    \caption{}
  \end{subfigure}  
  \begin{subfigure}[b]{.5\linewidth}
    \centering
    \begin{tikzpicture}[scale=1.2]
      \begin{scope}[every path/.style={{|[right]}-{|[left]}}]
        \foreach \l/\r/\y/\i/\c in {-.02/.02/1/1/, 0/3/2/2/,
          1/2/1/5/violet, 2.98/3.02/1/8/}{
          \draw[\c, thick] (\l-.02, \y/5) node[left, blue] {$\i$} to (\r+.02, \y/5);
        }
      \end{scope}
      \begin{scope}[every path/.style={{|[right]}-{|[left]}}, shift = {(1, -1)}]
        \foreach[count=\i from 3] \l/\r/\y/\c in {
          0.98/1.02/2/, 1/2/1/, 1.98/2.02/3/, 2/3/2/,
          2.98/3.02/1/}{
          \draw[\c, thick] (\l-.02, \y/5) node[left, blue] {$\i$} to (\r+.02, \y/5);
        }
        \draw[dashed, gray] (.7, 0) -- (3.2, 0);
      \end{scope}
      \draw[dashed, gray] (-.3, 0) -- (3.3, 0);

      \draw[dashed, thick, cyan] (1, 0.2) -- (2.2, -1) (2, 0.2) -- (4.6, -.8);
    \end{tikzpicture}
    \caption{}
  \end{subfigure}  
  \caption{(a) $G'$; (b) $G[M]$; and (c) interval representations for $G'$ (above) and $G[M]$ (below).}
  \label{fig:module-models}
\end{figure}
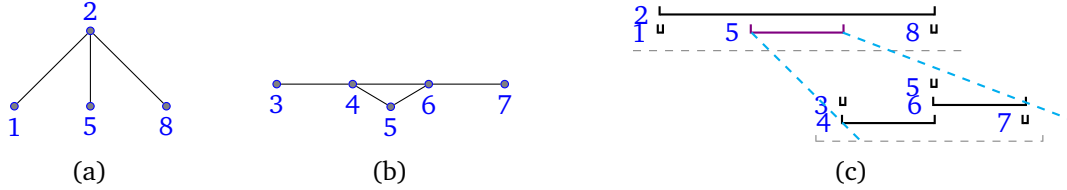

Modules and \textsc{lbfs} are closely related,\footnote{To find modules, both the algorithm of Hsu and Ma~\cite{hsu-99-recognizing-interval-graphs} for chordal graphs and the latest algorithm of Tedder et al.~\cite{tedder-08} for general graphs use \textsc{lbfs} as the workhorse.} and the aforementioned fact can be explained using \textsc{lbfs}.
Let $M$ be a module of a graph $G$.
Since \textsc{lbfs} selects vertices by adjacencies, vertices in $V(G)\setminus M$ have no impact on the ordering of vertices in $M$.
Before the first vertex of $\sigma|_{M}$ is visited, a snapshot either contains $M$ as a subset, or is disjoint from $M$.
The difference of labels of two vertices in $M$, if any, are completely in $M$.
Therefore, $\sigma|_{M}$ is an \textsc{lbfs} ordering of $G[M]$.
(We remark that vertices in $M$ are not necessarily consecutive in $\sigma$.)
On the other hand, only the first vertex $v$ of $\sigma|_{M}$ has impact on the ordering of other vertices in $M$: After the $(\sigma(v))$th iteration, if the smallest element in the difference of two labels is in $M$, then it has to be $v$.  Therefore, $\sigma|_{V(G')}$ is an \textsc{lbfs} ordering of $G'$.  If they are both interval orderings, then $\sigma$ is an interval ordering of $G$.

Most snapshots in an \textsc{lbfs} ordering are not modules.  Since a snapshot $S$ cannot be distinguished by vertices visited before $S$, we may view $S$ as a ``one-sided pseudo-module.''
A very nontrivial and crucial observation of Corneil et al.~\cite{corneil-09-lbfs-strucuture-and-interval-recognition} is that for a chordal graph, any non-module snapshot $S$ of an \textsc{lbfs} ordering $\sigma$ has the similar autonomous property; i.e., $\sigma|_{S}$ is an \textsc{lbfs} ordering of $G[S]$.
Since this is not completely obvious, let us briefly explain why it is true for interval graphs.  For this purpose, we are only concerned with those snapshots that are neither cliques nor modules: Any ordering of a clique $S$ is an \textsc{lbfs} ordering of $G[S]$, and we have discussed modules above.
We say that a vertex $v\not\in S$ \emph{splits} a set $S$, or that $v$ is a \emph{splitter} of $S$, if $\emptyset \subset N(v)\cap S \subset S$.
A set of vertices is a module if and only if it does not have any splitter.
Let $S$ be a snapshot of an \textsc{lbfs} ordering $\sigma$, and let us fix an interval representation $\mathcal{I}$ for $G$.  If $S$ is neither a clique nor a module, then $N_{\sigma}(S)$ is a clique, and $S$ are precisely intervals intersecting $\bigcap_{v\in N_{\sigma}(S)} I(v)$.  The splitters of $S$ form one or two cliques, and their intervals intersect the two ends of $\bigcap_{v\in N_{\sigma}(S)} I(v)$; e.g., the snapshot $S_{\sigma_2}(20)$, which is $\{6, 7, 9, 10, \ldots, 20\}$ of $\sigma_2$, has two splitters, $5$ and $21$.
  If an \textsc{lbfs} ordering of $G$ visits a splitter of $S$ before $S$ itself, e.g., $\sigma_1^{+}$, $\sigma_2^{+}$, or $\sigma_3^{+}$, then $S$ is visited from one end to the other.  Otherwise, no splitter of $S$ can be visited before the vertices in $S$ have been finished.  In either case, the ordering of $S$ is decided by $G[S]$ itself.

Combining these observations on modules and snapshots, we have the following theorem, which consumes all the modules of $G$ and all the snapshots of all possible \textsc{lbfs} orderings of $G$.

\begin{theorem}[The interval \textsc{lbfs} theorem \cite{corneil-09-lbfs-strucuture-and-interval-recognition}]
  \label{thm:chordal-lbfs}
  Let $G$ be an interval graph and let $\sigma$ be an \textsc{lbfs} orderings of $G$.  For any snapshot $S$ and any module $M$ of $G[S]$, the sub-ordering $\sigma|_{M}$ is an \textsc{lbfs} ordering of $G[M]$.
\end{theorem}

The execution of \textsc{lbfs} only considers adjacencies in one direction, which are not sufficient to tell whether a snapshot is a module or not.
We want to conduct multiple sweeps of  \textsc{lbfs}, and use information gleaned from the previous runs to decide whether a snapshot $S$ is a module of $G$, and more importantly, find a splitter of $S$ to orient $G[S]$, if $S$ is not a module of $G$.

We use the examples in Figure~\ref{fig:example-orderings} to motivate the main idea of the algorithm.
Suppose that $\sigma$ is an \textsc{lbfs} ordering of an interval graph $G$, and that $S$ is a non-module snapshot of $G$.
In general, the splitters of $S$ can sandwich $S$ from both sides, e.g., $S_{\sigma_2}(20)$.
A nontrivial observation is that if $\sigma$ is an \textsc{lbfs}$^{+}$ ordering,
then all the splitters of $S$ are at the same side of $S$.
In this case, we do not need to worry if $S$ is a clique either.
Now let $v$ be a splitter of non-clique non-module snapshot $S$.
Note that $v$ has to be after $S$ in $\sigma$; otherwise $S$ cannot be a snapshot.
Inspired by the algorithm for unit interval graphs, one may expect that in $\textsc{lbfs}^{+}(G, \sigma)$, vertex $v$ will be visited before $S$, thereby telling $S$ apart.
As shown by $\sigma_1$ and $\sigma_1^{+}$ in Figure~\ref{fig:example-orderings}, this natural idea is not true for interval graphs; note that $\sigma_1$ is $\textsc{lbfs}^{+}(G, \sigma')$ for
\[
  \sigma': 22, 4, 3, 2, 5, 6, 7, 8, 9, 18, 17, 12, 14, 13, 11, 15, 16, 10, 19, 20, 21, 1.
\]
In particular, the highlighted snapshot $\{6, 7, 9, 10, \ldots, 19\}$ of $\sigma_1$ has a unique splitter, vertex $5$, which is after this snapshot in both $\sigma_1$ and $\sigma_1^{+}$.

 Let us put the six orderings in Figure~\ref{fig:example-orderings} under a closer scrutiny.  We use $S^{\star}$ to denote the vertex set $\{6, 7, 9, 10, \ldots, 19\}$, whose only splitter is vertex $5$.
 In all the six orderings, $S^{\star}$ is a snapshot, and its only splitter is after $S^{\star}$.  Note that this is decided by vertices before $S^{\star}$ in $\sigma_i$, and cannot be changed by re-arranging vertices in $S^{\star}$.
 What we need to do is to find a way to 
 use vertex $5$ to orient $S^{\star}$ even if it is after $S^{\star}$.
 Since $G[S^{\star}]$ is prime, vertices $6$ and $19$ are its only two end vertices.
Since $\sigma_i, i = 1, 2, 3$, starts from vertex $1$ and ends at vertex $22$, if $\sigma_i^{+}$, which starts from vertex $22$ and ends at vertex $1$, is a correct interval ordering of $G^{\star}$, then $\sigma_i^{+}$ should be ``from right to left.''
In particular, $\sigma_i^{+}|_{S^{\star}}$ needs to start from $19$, which in turn requires that $\sigma_i|_{S^{\star}}$ start from $6$.
Among the three orderings $\sigma_i$, only $\sigma_3$ satisfies this condition.
We may informally say that $S^{\star}$ is ``anchored from the left'' in $\sigma_3$.
One may also note that the snapshot $\{11, 13,  14,  15, 16\}$, of which the only
splitter is vertex $10$, is also ``anchored from the left'' in $\sigma_3$.
On the other hand, the snapshot $\{3, 4, \ldots, 20\}$ of $\sigma_3$ is  ``anchored from the right'' by vertices $21$ and $22$, which does not make a problem because $\sigma_3^{+}$ starts from vertex $22$, and $\{3, 4, \ldots, 20\}$ is not a snapshot of $\sigma_3^{+}$.  As the reader may easily check, $\sigma_3^{+}$ is actually an interval ordering of $G^{\star}$.

These observations motivated the key concept of Corneil et al.~\cite{corneil-09-lbfs-strucuture-and-interval-recognition} and the main observation of Li and Wu~\cite{li-14-lbfs-interval-recognition}.
Let $\sigma$ be an \textsc{lbfs} ordering of an interval graph $G$, and let $S$ be a snapshot of $\sigma$.
We say that a vertex $v\in S$ is \emph{exposed} (from $S$) in $\sigma$ if $N(v)\setminus (N_{\sigma}(S)\cup S)\ne \emptyset$; i.e., if some neighbor of $v$ is after $S$ in $\sigma$.
By the definition of \textsc{lbfs}, if $v$ is adjacent to some splitter of $S$, then it must be exposed; on the other hand, if $S$ is not a clique, then by the perfect elimination theorem (Theorem~\ref{thm:peo}), every neighbor after $S$ in $\sigma$ is a splitter of $S$.
In the example above, vertex $6$ is the only exposed vertex from the snapshot $S^{\star}$ in $\sigma_i, i= 1, 2, 3$, and vertex $11$ is the only exposed vertex from the snapshot $\{11, 13,  14,  15, 16\}$ in $\sigma_1$ and $\sigma_3$.
We say that an \textsc{lbfs} ordering $\sigma$ of an interval graph is \emph{well-anchored} if for any snapshot $S$ of $\sigma$, the sub-ordering $\sigma|_S$ starts from
\begin{enumerate}[({A}1)]
\item an exposed vertex from $S$ if one exists; or
\item an end vertex of $G[S]$ otherwise.
\end{enumerate}
Note that since $V(G)$ itself is a snapshot and a trivial module, a well-anchored ordering always starts from an end vertex of $G$.
We urge the reader to verify that $\sigma_3$ is indeed well-anchored.
As demonstrated above, 
if a snapshot $S$ of a well-anchored ordering $\sigma$ is again a snapshot in \textsc{lbfs}$^{+}(G, \sigma)$, then the first vertex of $S$ can help us to orient $S$.  As we will show below, if $\sigma$ is well-anchored, then $\textsc{lbfs}^{+}(G, \sigma)$ must be an interval ordering of $G$.

However, it is not immediate clear how to produce in linear time an \textsc{lbfs} ordering that is well-anchored.
Indeed, it is already very challenging to decide whether an exposed vertex exists in a snapshot, which is equivalent to testing whether a set is a module.  This difficulty is manifested by the fact that an \textsc{lbfs} checks adjacencies with visited vertices, while a splitter of $S$ and an exposed vertex from $S$, if they exist, are both unvisited at the moment choosing the first vertex of $S$.  
The crucial observation is that we do \emph{not} really need to know whether an exposed vertex exists, and it suffices to make sure that the first vertex is exposed if any one is.
Let $G$ be an interval graph, and let $\tau^{+}$ be an \textsc{lbfs}$^{+}$ ordering of $G$.  We conduct a new \textsc{lbfs} $\pi$ from the last vertex of $\tau^{+}$.  Whenever a proper snapshot $S$ with more than one vertex is met, we proceed as follows.
In the first case, some vertex $x\in N(S)$ was before $S$ in $\tau^{+}$ but after $S$ in $\pi$.
If $S\subseteq N(x)$, then every vertex in $S$ is exposed.  Otherwise, $x$ is a splitter of $S$, and then by the procedure \textsc{lbfs}, the first vertex of $\sigma|_{S}$ must be adjacent to the first splitter of $S$ in $\tau^{+}$ (not necessarily $x$).  Therefore, the first vertex of $\sigma|_{S}$ is exposed.
In the rest, every vertex $x\in N(S)$ that is after $S$ in $\pi$ is after $S$ in $\tau^{+}$ as well.  (Note that $x$ is necessarily a splitter of $S$ when $S$ is not a clique, but as said we would not bother ourselves with whether this is true.)  It suffices to find a vertex in $S$ that has a neighbor after $S$ in $\tau^{+}$, and in the absences of such a vertex, we take the last vertex of $\tau^{+}|_{S}$.  The procedure is summarized in Figure~\ref{fig:alg-lbfs*}.
It is worth stressing again that the procedure does not calculate the maximal cliques or the modules explicitly.  The reader may verify that $\textsc{lbfs}^{\uparrow}(G, \sigma_1^{+})$ is precisely $\sigma_3$ in Figure~\ref{fig:example-orderings}.

\begin{figure}[h!]
  \centering 
  \begin{tikzpicture}
    \path (0,0) node[text width=.85\textwidth, inner sep=10pt] (a) {
      \begin{minipage}[t!]{\textwidth}
        \begin{tabbing}
        Procedure $\textsc{lbfs}^{\uparrow}(G, \tau^{+})$.
        \\
          Output: \= \kill
        Input: \>A graph $G$, and an \textsc{lbfs}$^{+}$ ordering $\tau^{+}$ of $G$.
        \\
        Output: A well-anchored ordering of $G$ if $G$ is an interval graph.
      \end{tabbing}

      \begin{tabbing}
          AAA\=Aaa\=aaa\=Aaa\=MMMMMAAAAAAAAAAAA\=A \kill
          1.\> renumber the vertices such that $\tau^{+}(v_i) = i$ for all $i = 1, \ldots, n$;
          \\
          2.\> {\bf for} $i = 1, \ldots, n$ {\bf do}
          \\
          2.1.\>\> $S\leftarrow$ unvisited vertices with the lexicographically largest label;
          \\
          2.2.\>\> $v_p\leftarrow$ {the first vertex of  $\tau^{+}|_{S}$};
          \\
          2.3.\>\> $v_q\leftarrow$ {the last vertex of  $\tau^{+}|_{S}$};
          \\
          2.4.\>\> \textbf{if} there exists $\ell < p$ such that $v_\ell \in N(v_p)$ and $\pi(v_{\ell})$ is unset \textbf{then} 
          \\
          \>\>\> $\pi(v_p)\leftarrow i$;
          \\
          2.5.\>\> \textbf{else if} there exist $v_\ell \in S$ and $v_r\in N(v_{\ell})$ such that $r > q$ \textbf{then} 
          \\
          \>\>\> $\pi(v_{\ell})\leftarrow i$;  
          \\
          2.6.\>\> \textbf{else} $\pi(v_q)\leftarrow i$; \>\>\>\codecomment{In this case $S$ is a module of $G$.}
          \\
          2.7.\>\> \textbf{for each} unvisited neighbor of $v$ \textbf{do} 
          \\
          \>\>\> add $i$ to $\mathrm{label}(v)$;
          \\
          3.\> \textbf{return} $\pi$.
        \end{tabbing}
      \end{minipage}
    };
    \draw[draw=gray!60] (a.north west) -- (a.north east) (a.south west) -- (a.south east);
  \end{tikzpicture}
  \caption{The procedure for producing a well-anchored ordering.}
  \label{fig:alg-lbfs*}
\end{figure}

\begin{lemma}[\cite{li-14-lbfs-interval-recognition}]
  \label{lem:lbfs-star}
  Let $\tau^{+}$ be an \textsc{lbfs}$^{+}$ ordering of an interval graph $G$.  Then \textsc{lbfs}$^{\uparrow}(G, \tau^{+})$ is a well-anchored \textsc{lbfs} ordering of $G$.
\end{lemma}
\begin{proof}
  Note that at the beginning of the procedure, we have renumbered the vertices such that $\tau^{+}(v_i) = i$ for all $i = 1, \ldots, n$.
  Let $\pi = $\textsc{lbfs}$^{\uparrow}(G, \tau^{+})$.
  It is clear from the procedure that the $i$th vertex of $\pi$, chosen in step~2.4, 2.5, or 2.6 of the $i$th iteration is in $S$, and thus $\pi$ is an \textsc{lbfs} ordering of $G$.  
  We need to show that for every $i = 1, \ldots, n$, the $i$th snapshot of $\pi$ satisfies the conditions in the definition of well-anchored orderings.  Let us fix an $i$ and suppose that $\pi(v_s) = i$, and let us use  $S$ to denote the $i$th snapshot, i.e., $S_{\pi}(v_s)$.

  In the first case, the procedure enters step~2.4.  There is a vertex $v_{\ell} \in N(v_p)\setminus N_{\pi}(v_p)$ with $\ell < p$.  The selection of $p$ and the fact $\ell < p$ imply $v_{\ell}\not\in S$; i.e., $v_{\ell}$ is after $S$ in $\pi$.
  Thus, $v_p$, the first vertex of $\pi|_S$, is exposed from $S$ in $\pi$.
  Henceforth we may assume that the condition of step~2.4 is not true.
  Then
  \begin{equation}
    \label{eq:1}
    N_{\tau^{+}}(v_p)\subseteq N_{\pi}(v_p)\setminus S = N_{\pi}(S) \subseteq N_{\pi}(v_q),
  \end{equation}
  and hence $v_q$ is also in the snapshot $S_{\tau^{+}}(v_p)$ by the rule of \textsc{lbfs}.
  The selection of $p$ and $q$ then implies
    $S\subseteq S_{\tau^{+}}(v_p)$.
    By the interval \textsc{lbfs} theorem (Theorem~\ref{thm:chordal-lbfs}), $\tau^{+}|_{S}$ is an \textsc{lbfs} ordering of $G[S]$, of which $v_q$ is the last vertex.  Thus, $v_q$ is an end vertex of $G[S]$.
    On the other hand, since $\tau^{+}$ is an \textsc{lbfs}$^{+}$ ordering, $v_p$ is an end vertex of $G[S_{\tau^{+}}(v_p)]$, hence simplicial in $G[S_{\tau^{+}}(v_p)]$.  Since $G[S]$ is an induced subgraph of $G[S_{\tau^{+}}(v_p)]$, the vertex $v_p$ is also simplicial in $G[S]$.

    In the second case, $S$ is a module.  If there exists a vertex in $S$ that is exposed from $S$, then any vertex in $S$ is exposed from $S$; thus
$S$ satisfies condition (A1).
Now that no vertex in $S$ is exposed from $S$, we need to make sure that $S$ satisfies condition (A2); i.e., $v_s$ is an end vertex of $G[S]$.  Again, this is trivial when $S$ is a clique, of which every one is an end vertex.  Now suppose that $S$ is not a clique, then a vertex $v_r\in N(S)$ with $r > q$ would contradict the perfect elimination theorem (Theorem~\ref{thm:peo}).  Thus, the procedure will skip step~2.5 and enter step~2.6, and then $v_s = v_q$.
  
In the rest, $S$ is not a module of $G$.  Let $v_r$ be a splitter of $S$; since $S\subseteq S_{\tau^{+}}(v_p)$, we must have $r > p$ by the rule of \textsc{lbfs}.
  We show by contradiction that $r > q$.  Suppose that $r < q$, then $v_r\in S_{\tau^{+}}(v_p)$, and by the definition of snapshots, $N_{\tau^{+}}(v_p) \subseteq N(v_r)$.
 If $S$ is a clique, then from $S\not\subseteq N[v_r]$ and the fact that $v_p$ is simplicial in $G[S]$ it can be inferred that $v_r$ and $v_p$ are not adjacent.  As a result, the $(p+1)$st snapshot of $\tau^{+}$, which is $S_{\tau^{+}}(v_p)\cap N(v_p)$, contains $v_q$ but not $v_r$, a contradiction to $r < q$.

 Now suppose that $S$ is not a clique.
 Since an \textsc{lbfs} ordering $\tau^{+}|_{S}$ of $G[S]$ starts with $v_p$ and ends with $v_q$, both simplicial vertices in $G[S]$, we can conclude that $v_p$ and $v_q$ are not adjacent.
Since $v_r$ is after $S$ in $\pi$, by the rule of \textsc{lbfs}, there exists a vertex $v_j\in N_{\pi}(S)\setminus N(v_r)$.  By the definition of snapshots, $v_j$ is adjacent to both $v_p$ and $v_q$.
Since $v_j$ is adjacent to a proper and nonempty subset of $S_{\tau^{+}}(v_p)$, it is not in $N_{\tau^{+}}(v_p)$.
By the perfect elimination theorem (Theorem~\ref{thm:peo}), we can infer from $v_p v_q \not\in E(G)$ that (i) $v_r$ is not adjacent to both $v_p$ and $v_q$ because $v_p <_{\pi} v_q <_{\pi} v_r$; and (ii) $j< q$ because $v_j$ is adjacent to both $v_p$ and $v_q$.
Thus, $v_j$ is in $S_{\tau^{+}}(v_p)$ as well.
Since $v_p$ is simplicial in $G[S_{\tau^{+}}(v_p)]$ and $v_j$ is adjacent  to $v_p$ but not $v_r$, we have $v_r\not\in N(v_p)$.  Since $r < q$, there must be some $v_k\in N(v_r)\setminus N(v_q)$ with $k< j$; moreover, $v_k$ is in $S_{\tau^{+}}(v_p)$ because both $v_r$ and $v_q$ are both in $S_{\tau^{+}}(v_p)$.  Thus, $p < k < j$.  As a result, both $v_j$ and $v_k$ are adjacent to $v_p$, which further implies that $v_k$ and $v_j$ are adjacent because $v_p$ is simplicial in $G[S_{\tau^{+}}(v_p)]$.  Hence $\{v_p, v_k, v_j, v_r, v_q\}$ induces a bull, with edges $v_k v_j, v_k v_r, v_p v_j, v_q v_j, v_p v_k$.  But $v_p$, an end vertex of $G[S_{\tau^{+}}(v_p)]$, has degree two in this bull, contradicting Lemma~\ref{lem:bull}.

We have thus concluded that $r > q$.
  Since $\emptyset \subset N(v_r)\cap S \subset S$, it follows that $v_r\not\in N_{\pi}(S)$.  Therefore, the condition of step~2.5 must be true.  (One may note that this means the procedure can reach step~2.6 only when $S$ is a module of $G$.)
  Then $v_{s}$ is exposed from $S$ in $\pi$ because it is adjacent to some splitter of $S$.  This concludes the proof.  
\end{proof}

We are ready to present the main algorithm for recognizing interval graphs in Figure~\ref{fig:alg-lbfs*}, which is very straightforward now.  Note that we use consistent symbols in the procedure \textsc{lbfs}$^{\uparrow}$ and the main algorithm.
In this rest of this section, we always use $\tau$, $\tau^{+}$, $\pi$, and $\pi^{+}$ to denote the \textsc{lbfs} orderings of $G$ produced by the first four steps of this algorithm.

\begin{figure}[h!]
  \centering 
  \begin{tikzpicture}
    \path (0,0) node[text width=.85\textwidth, inner sep=10pt] (a) {
      \begin{minipage}[t!]{\textwidth}
        \begin{tabbing}
          Output: \= \kill
          Input: \> A connected graph $G$.
          \\
          Output: Whether $G$ is an interval graph.
        \end{tabbing}        

        \begin{tabbing}
          AAA\=Aaa\=aaa\=Aaa\=MMMMMMAAAAAAAAAAAAA\=A \kill
          1.\> $\tau\leftarrow$ an \textsc{lbfs} ordering of $G$;
          \\
          2.\> $\tau^{+}\leftarrow \textsc{lbfs}^{+}(G, \tau)$;
          \\
          3.\> $\pi\leftarrow \textsc{lbfs}^{\uparrow}(G, \tau^{+})$;
          \\
          4.\> $\pi^{+}\leftarrow \textsc{lbfs}^{+}(G, \pi)$;
          \\
          5.\> \textbf{if} $\pi^{+}$ is an interval ordering of $G$ \textbf{then return} ``yes'';
          \\
          6.\> \textbf{else return} ``no.''
        \end{tabbing}
      \end{minipage}
    };
    \draw[draw=gray!60] (a.north west) -- (a.north east) (a.south west) -- (a.south east);
  \end{tikzpicture}
  \caption{The recognition algorithm for interval graphs.}
  \label{fig:alg-main}
\end{figure}

Before the formal statement of the implication of well-anchored orderings, let us
again use $\sigma_3$ and $\sigma_3^{+}$ in Figure~\ref{fig:example-orderings} for an illustration.  They are reproduced below, with extra marks.
The purpose of producing a well-anchored ordering $\pi$ is to force the exposed vertex to be visited by $\pi^{+}$ as early as possible to ``anchor the set $S$'' correctly in $\pi^{+}$.
The four proper snapshots $S$ of $\sigma_3^{+}$ that are not cliques start from $21$, $19$, $16$, and $14$, as denoted by brackets, and the first non-universal vertex in $\pi|_{S}$ are, respectively, $2$, $6$, $11$, and $11$, as shown in parentheses.

\begin{align*}
  \sigma_3&: 1, (2), 4, 20, 8, (6), 7, 9, 18, 17, 12, (11), 13, 15, 14, 16, 10, 19, 5, 3, 21, 22.
  \\
  \sigma_3^{+}&: 22, 4, [21, 20, 8, 2, [19, 18, 17, 9, 12, [16, 15, 13, [14, 11]], 10, 7, 6], 5, 3], 1.
\end{align*}

\begin{lemma}[\cite{li-14-lbfs-interval-recognition}]
  \label{lem:final-snapshot}
  Let $\pi$ be a well-anchored \textsc{lbfs} ordering of an interval graph $G$, and let $S$ be a non-clique snapshot of $\pi^{+}$.
  If the first vertex of $\pi|_{S}$ is not exposed from $S$ in $\pi^{+}$, then $S$ is a module of $G$.
\end{lemma}
\begin{proof}
  For the proof we may renumber the vertices in a way that $\pi(v_i) = i$ for all $i = 1, \ldots, n$.
  Suppose that $v_p$ the first vertex of $\pi|_{S}$.  
  Since $v_p$ is not exposed from $S$ in $\pi^{+}$, we have $N(v_p)\setminus S = N_{\pi^{+}}(S)$.
  By the selection of $p$,
  \begin{equation}
    \label{eq:6}
    N_{\pi}(v_p) \subseteq N(v_p)\setminus S = N_{\pi^{+}}(S),
  \end{equation}
    and thus every vertex in $N_{\pi}(v_p)$ is adjacent to all vertices in $S$ because $S$ is a snapshot of $\pi^{+}$.
    From \eqref{eq:6} we can also conclude $N_{\pi}(v_p)\subseteq N(v_j)$ for every $v_j\in S\setminus \{v_p\}$.
  The selection of $p$ implies $j > p$, and thus $v_j\in S_{\pi}(v_p)$ by the rule of \textsc{lbfs}. 
In other words, $S\subseteq S_{\pi}(v_p)$, and then $S_{\pi}(v_p)$ is not a clique because $S$ is not a clique,.

  By the definition of snapshots, every vertex $v_i\in N(v_p)\setminus S=N_{\pi^{+}}(S)$ is adjacent to all vertices in $S$, and then by the perfect elimination theorem (Theorem~\ref{thm:peo}), $v_i$ cannot be after $S$ in $\pi$.
  Thus, $v_p$ is not exposed from $S_{\pi}(v_p)$ in $\pi$.
  Since $\pi$ is a well-anchored ordering, $S_{\pi}(v_p)$ has to be a module of $G$, and $v_p$ is an end vertex of $G[S_{\pi}(v_p)]$.
    By the flipping lemma (Lemma~\ref{lem:flipping}), $\pi^{+}|_{S_{\pi}(v_p)}$ ends with $v_p$.
    Thus, $v_{i}\not\in S_{\pi}(v_p)$ for every vertex $v_{i}$ with $v_p <_{\pi^{+}} v_i$; since $S_{\pi}(v_p)$ is a module of $G$, if $v_i$ is adjacent to $S\subseteq S_{\pi}(v_p)$, then $v_i$ is adjacent to all the vertices in $S_{\pi}(v_p)$, which contradicts the perfect elimination theorem (Theorem~\ref{thm:peo}).
  Therefore, $N(S)= N_{\pi^{+}}(v_p)$, which means that $S$ is a module of $G$.
\end{proof}

Before the main theorem of this section, we need two more simple properties on modules.  The first property extends the similar statement on general \textsc{lbfs} ordering, and is quite natural.
Recall that \textsc{lbfs}$^{+}$ is deterministic: $\textsc{lbfs}^{+}(G, \sigma)$ is unique for any graph $G$ and any \textsc{lbfs} ordering $\sigma$ of $G$.
\begin{lemma}\label{lem:interval-modues}
  Let $\pi$ be a well-anchored \textsc{lbfs} ordering of an interval graph $G$, and let $\pi^{+} = \textsc{lbfs}^{+}(G, \pi)$.  For any module $M$ of $G$, the sub-ordering $\pi|_{M}$ is a well-anchored ordering of $G[M]$, and $\pi^{+}|_{M}= \textsc{lbfs}^{+}(G, \pi|_{M})$.
\end{lemma}
\begin{proof}
  By the interval \textsc{lbfs} theorem (Theorem~\ref{thm:chordal-lbfs}), $\pi|_{M}$ is an \textsc{lbfs} ordering of $G[M]$.
  Moreover, $N_{\pi|_{M}}(v) = N_{\pi}(v)\cap M$ and $S_{\pi|_{M}}(v) = S_{\pi}(v)\cap M$ for every vertex $v$ in $M$.
  It is easy to use definition to verify that $S_{\pi|_{M}}(v)$ is a module of $G$ if and only if  $S_{\pi|_{M}}(v)$ is a module of $G[M]$ as well.
  If $S_{\pi|_{M}}(v)$ is a module of $G[M]$, then the fact that the first vertex of $S_{\pi|_{M}}(v)$ is an end vertex of $G[S_{\pi|_{M}}(v)]$ follows from the definition of well-anchored orderings itself.  In the rest, $S_{\pi|_{M}}(v)$ is not a module of $G[M]$.  Note that then $M$ is not a clique.
  By the perfect elimination theorem (Theorem~\ref{thm:peo}), there is no vertex in $N(M)$ that is after $M$ in $\pi$.
  Since $M$ is a module, a splitter of $S_{\pi|_{M}}(v)$ is in $M$, and thus cannot be in $S_{\pi}(v)$.
  Thus, $S_{\pi}(v)$ is not a module of $G$, and by the definition of well-anchored orderings, $v$ is exposed from $S_{\pi}(v)$ in $\pi$.
  Now that $v$ is adjacent to some vertex in $M$ that is after $S_{\pi|_{M}}(v)$ in $\pi$, it is exposed from $S_{\pi|_{M}}(v)$ in $\pi|_{M}$.
  Thus, $\pi|_{M}$ is well-anchored.
  
  For the proof of $\pi^{+}|_{M}= \textsc{lbfs}^{+}(G, \pi|_{M})$, we may renumber the vertices in a way that $\textsc{lbfs}^{+}(G, \pi|_{M}) = \langle v_1, v_2, \ldots, v_{|M|}\rangle$, with vertices in $V(G)\setminus M$ arbitrarily from $v_{|M|+1}$ to $v_{n}$.  
  Suppose for contradiction $\pi^{+}|_{M}\ne \textsc{lbfs}^{+}(G, \pi|_{M})$, and let $i$ be the smallest number such that $\pi^{+}|_{M}(v_j) = i$, with $i \ne j$.  Let $S = S_{\pi^{+}}(v_j)$.  By the selection of $i$ and the definition of modules, the set $S\cap M$ is precisely the vertices in the $i$th snapshot of $\textsc{lbfs}^{+}(G, \pi|_{M})$.  By the rule of \textsc{lbfs}$^{+}$, the first vertex of $\pi^{+}|_{S}$ should be the last vertex of $\pi|_{S}$, which is $v_i$.  We have a contradiction.
\end{proof}

We need a constructive version of Theorem~\ref{thm:prime}.  The following lemma implies Theorem~\ref{thm:prime}.
\begin{lemma}\label{lem:modules}
  Let $G$ be an interval graph with $\ell$ maximal cliques.  Let $K_1$, $K_{2}$, $\ldots$, $K_\ell$ and
  $K_{b(1)}$, $K_{b(2)}$, $\ldots$, $K_{b(\ell)}$
  be two different clique paths of $G$. 
  If there are $p$ and $q$ with $1\le p < q<\ell$ such that $b(p) =1$ and $b(q) =\ell$, 
  then
  
  \[
    \bigcup_{j = b(\ell)}^{\ell} K_j
    \cup \bigcup_{i=q}^{\ell} K_{b(i)}
    \setminus (K_{b(\ell)}\cap K_\ell)
  \]
  is a nontrivial module of $G$.
\end{lemma}
\begin{proof}
  Let
  \[
    J= \{b(\ell), b(\ell)+1, \ldots, \ell\} \cup\{b(q), b(q+1), \ldots, b(\ell)\}
  \]
  and
  $U = \bigcup_{j\in J} K_j \setminus (K_{b(\ell)}\cap K_\ell)$.  We show that $N(v)\setminus U =  K_{b(\ell)}\cap K_\ell$ for every $v\in U$.  Since $v$ is in a clique that is between $K_{b(\ell)}$ and $K_\ell$, in at least one of the two clique paths.  It follows from the definition of clique paths that $K_{b(\ell)}\cap K_\ell\subseteq N(v)$.  It remains to show that $N(v)\setminus U\subseteq  K_{b(\ell)}\cap K_\ell$.
  If $j\in J$ for every maximal clique $K_j$ containing $v$, then
  \[
    N[v] \subseteq \bigcup_{j\in J} K_j = U\cup (K_{b(\ell)}\cap K_\ell),
  \]
  and thus
  $N(v)\setminus U\subseteq  K_{b(\ell)}\cap K_\ell$.  Now suppose that there exists $j\not\in J$ such that $v\in K_j$, then $j < {b(\ell)}$ and there is $k<q$ such that $b(k) =j$.  Then since $v$ can be found in both sides of $K_{b(\ell)}$ in the first clique path, it has to be in $K_{b(\ell)}$ as well.  For the same reason, $v\in K_\ell$.  But then $v$ is in $K_{b(\ell)}\cap K_\ell$, and should not be in $U$, a contradiction.  This concludes the proof.
\end{proof}

To prove the main lemma, we show that $\pi^{+}$ is consistent with some clique path $\mathcal{K}$ of $G$.
The main strategy is that if this is not true, then we can use Lemma~\ref{lem:modules} to identify a nontrivial module of $G$.  This cannot happen for a prime graph.
In general, however, $G$ might have modules.  If a module that is maximal in a certain sense is a counterexample, we work on this module only.\footnote{The reader who is familiar with modular decomposition may notice that we can assume that all the nontrivial modules are consistent with $\sigma$.  For our purpose, we do not need the full power of modular decomposition.}
Otherwise, by virtue of Lemma~\ref{lem:final-snapshot}, we can assume that for each $M$ of these modules, $\sigma|_{M}$ is consistent with the sub-path of $\mathcal{K}$ for $M$.  Again, if $\sigma$ is not consistent with $\mathcal{K}$, then we use Lemma~\ref{lem:modules} to identify a module of $G$, and we end with a similar contradiction.  
Note that for any non-clique module $M$ of $G$, the set $M\setminus U$, while $U$ is the set of universal vertices of $G[M]$, is also a module of $G$.

\begin{lemma}[\cite{li-14-lbfs-interval-recognition}]
  \label{lem:correctness-prime}
  Let $\pi$ be an \textsc{lbfs} ordering of an interval graph $G$.  If $\pi$ is well-anchored, then $\textsc{lbfs}^{+}(G, \pi)$ is an interval ordering of $G$.
\end{lemma}
\begin{proof}
  We may assume that $G$ is connected and has no universal vertices.    By Lemma~\ref{lem:interval-modues}, if $G$ is not connected, then we may work on its components one by one; if $G$ has a set $U$ of universal vertices, we may consider $G - U$.  We say that a non-clique module $M\subset V(G)$ is \emph{major} if $G[M]$ has no universal vertex and $M$ is maximal in this sense; i.e., the only module of $G$ that does not have universal vertices and properly contains $M$ is $V(G)$ itself.
  Since a major module $M$ is not a clique, $N(M)$ is a clique.
  
  We argue that two major modules are disjoint.  Suppose for contradiction that the intersection of two major modules $M_1$ and $M_2$ is not empty.  By definition neither of $M_1$ and $M_2$ is a subset of the other.  If there is no edge between $M_1\cap M_2$ and other vertices in  $M_1\cup M_2$, then by the definition of modules, there is no edge between $M_1\setminus M_2$ and $M_2\setminus M_1$ either.  In other words, the three parts, $M_1\setminus M_2$ and $M_2\setminus M_1$ and $M_1\cap M_2$, all comprise components of $G[M_1\cup M_2]$.  On the other hand, if a vertex in $V(G)\setminus (M_1\cup M_2)$ is adjacent to any vertex in $M_1$, then it is adjacent to all vertices in $M_1\cap M_2$, hence also to all vertices in $M_2$.  Therefore, $M_1\cup M_2$ is a module of $G$, and $M_1\cup M_2\ne V(G)$ because $G$ is connected.  We have thus a contradiction to that $M_1$ and $M_2$ are major modules.  Now that $M_1\cap M_2$ has at least one neighbor $x$ in $M_1\setminus M_2$, the vertex $x$ is adjacent to all vertices in $M_2$.  This further implies that every vertex in $M_2\setminus M_1$ is adjacent to all vertices in $M_1$.  But then at least one of $M_2\cap M_1$ and $M_2\setminus M_1$ is a clique, hence consisting of universal vertices in $G[M_2]$.  This contradicts that $M_2$ is a major module.  We end with the same contradiction if $x\in M_2\setminus M_1$ is adjacent to $M_1\cap M_2$.

  We argue that in any clique path of $G$, (a) maximal cliques containing vertices in a major module $M$ are consecutive, and (b) they can be replaced by any clique path of the subgraph induced by $N[M]$.  Suppose that $K_{1}$, $K_{2}$, $\ldots$, $K_{\ell}$ is a clique path of $G$, where ${p}$ and ${q}$ are the smallest and, respectively, largest indices such that $K_{p}, K_{q}\subseteq N[M]$.
  By the selection of $p$ and $q$, if $i < p$ or $i > q$, then $K_i\cap M=\emptyset$; otherwise $K_i\setminus M\subseteq N(M)$, and then $K_i\subseteq N[M]$.
  As a result,
  \begin{equation}
    \label{eq:5}
    K_{p}\cap K_{p-1}\subseteq K_{p}\cap K_{q} \text{ and } K_{q}\cap K_{q+1}\subseteq K_{p}\cap K_{q}
  \end{equation}
  when $p >1$ and when $q < \ell$ respectively.
  Now that $M\subseteq \bigcup_{i=p}^{q} K_{i} \setminus (K_{p}\cap K_{q})$, we have $q > p$ because $M$ is not a clique.
  We then show that $\bigcup_{i=p}^{q} K_{i} \setminus (K_{p}\cap K_{q})\subseteq M$; i.e., they are equivalent.  
  It suffices to show that $\bigcup_{i=p}^{q} K_{i} \setminus (K_{p}\cap K_{q})$ is actually a module of $G$ and no vertex is universal in it.
  The first follows from that $K_{p}\cap K_{q}\subseteq N(v) \subseteq \bigcup_{i=p}^{q} K_{i}$ for every vertex $v\in K_{i} \setminus (K_{p}\cap K_{q})$ with $p \le i \le q$.  Any vertex in $\bigcup_{i=p}^{q} K_{i} \setminus (K_{p}\cap K_{q})$ is absent from at least one of $K_{p}$ and $K_{q}$, and hence cannot be universal.
  The fact (b) follows from (a), \eqref{eq:5}, and the definition of clique paths.
  
  Let $\pi^{+} = \textsc{lbfs}^{+}(G, \pi)$.  
  For the rest of the proof we renumber the vertices in $G$ in a way that $\pi^{+}(v_i) = i$ for $i = 1, \ldots, n$.
  We prove by contradiction that there is a clique path of $G$ that is consistent with $\pi^{+}$.
  We may assume that for each major module $M$ of $G$, there exists a clique path $\mathcal{K}_{M}$ that is consistent with $\pi^{+}|_{M}$; otherwise, by Lemma~\ref{lem:interval-modues}, we may focus on $G[M]$ and its orderings $\pi|_{M}$ and $\pi^{+}|_{M}$.
  We fix a clique path $\mathcal{K}$ for $G$ such that for every major module $M$, the sub-path of $\mathcal{K}$ for $G[M]$ is consistent with $\pi^{+}|_{M}$; in particular, the first vertex in $\pi^{+}|_{M}$ is in the first maximal clique of $G[N[M]]$.
  Note that it exists because the two properties on major modules we proved above.  Starting from an arbitrary clique path of $G$, for each major module $M$, we can replace the sub-path for $N[M]$ by one consistent with $\pi^{+}|_{M}$.  

  Suppose for contradiction that $\pi^{+}$ is not consistent with $\mathcal{K}$.
  There exists a pair of vertices $v_p$ and $v_q$ in $G$ such that $p < q$ but $\lp{v_p} > \lp{v_q}$.  Let them be chosen such that $p$ is the minimum and $\rp{v_q}$ is minimum with respect to this fixed $p$.
  We denote by
  \[
    r=\lp{v_q}\text{ and }t = \lp{v_p},
  \]
  and let
  \[
    S = S_{\pi^{+}}(v_p)\text{ and }X=N_{\pi^{+}}(v_p).
  \]
  By the selection of $p$, for all $i = 1, \ldots, p - 1$, we have $\lp{v_i} \le \lp{v_q}$.
  Therefore, $X \subseteq N(v_q)$, and $v_q$ is in the snapshot $S$.
  We argue that $v_q$ is a simplicial vertex of $G$.  Otherwise, by the definition of clique paths, there exists a vertex $v_j$ with $\rp{v_j} = \lp{v_q} < \rp{v_q}$, and then $j < p$ by the selection of $v_p$.  But since $v_j\in N(v_q)\setminus N(v_p)$, \textsc{lbfs} should visit $v_q$ before $v_p$.
  Now that $v_q$ is simplicial,
  from $t > r$ we can conclude that $v_p$ and $v_q$ are not adjacent.  As a consequence, $X$ is a clique.  We also argue that $S$ is not a module of $G$.  If $S$ is a module, then there is a major module $M$ such that (a) all the non-universal vertex of $G[S]$ are in $M$; and (b) a universal vertex of $G[S]$ is either adjacent to all vertices in $M$ or in $M$.  But then the existence of $v_p$ and $v_q$ would contradict the assumption that $\pi^{+}|_{M}$ is consistent with maximal cliques in $N[M]$.

  Let $K_{r'}$ and $K_{t'}$ be the first and, respectively, the last maximal cliques in $\mathcal{K}$ that is a subset of $S\cup X$.
  For each vertex $v_i$ that is adjacent to all vertices in $S$, if $i \ge p$, then $v_i\in S$; if $i < p$, then $v_i\in X$.
  Therefore, no vertex in $V(G)\setminus (S\cup X)$ can be adjacent to all vertices in $S$ and $X$.  
  As a result, every maximal clique $K$ of $G[S\cup X]$, which contains $X$ as a subset, is a maximal clique of $G$.  Moreover, every maximal clique $K$ of $G$ wih $X\subseteq K$ and $K\cap S\ne \emptyset$ is a subset of $S\cup X$.  In particular, $K_r, K_t\subseteq S\cup X$.  Thus, $r' \le r < t \le t'$, and by assumption, every vertex in $K_1$, $\ldots$, $K_{r'-1}$ are before $v_p$ in $\pi^{+}$.
  Another consequence is that $K_{j}\setminus X \subseteq S$ for all $j$ with $r'\le j \le t'$.  Therefore, any vertex after $S$ in $\pi^{+}$ is in a clique $K_{j}$ with $j > t'$.  By the definition of clique paths, we can conclude that all the exposed vertices from $S$ in $\pi^{+}$ belong to $K_{t'}$.  

By the interval \textsc{lbfs} theorem (Theorem~\ref{thm:chordal-lbfs}), $v_{p}$ is an end vertex of $G[S]$.
 Therefore, there exists a clique path $\mathcal{K}'$ for $G[S\cup X]$ in which $K_{t}$ is at one end.  We may assume without loss of generality that $K_{t}$ is the first of $\mathcal{K}'$.
 On the other hand, let $v_z$ be the first vertex of $\pi|_{S}$ that is not universal in $G[{S}]$.
We have seen that $S$ is not a module of $G$, and thus $v_z$ is exposed from $S$ in $\pi^{+}$ by Lemma~\ref{lem:final-snapshot}.
By the interval \textsc{lbfs} theorem (Theorem~\ref{thm:chordal-lbfs}), $\pi|_{S}$ is an \textsc{lbfs} ordering of $G[S]$, and thus the first non-universal vertex $v_z$ and the last vertex $v_p$ cannot be adjacent.  Further, from $v_z\in K_{t'}$ we can conclude that $v_p\not\in K_{t'}$ and $t < t'$.
Thus, in the clique path $K_{r'}$, $K_{r'+1}$, $\ldots$, $K_{r}$, $K_{r+1}$, $\ldots$, $K_{t}$, $\ldots$, $K_{t'}$ for $G[S\cup X]$, the clique $K_{t}$ is not an end.

If $K_{r'}$ is before $K_{t'}$ in $\mathcal{K}'$, then by Lemma~\ref{lem:modules}, there is a module $M$ of $G[S\cup X]$ that contains all vertices in $\bigcup_{j=r'}^{t} K_{j}\setminus (K_{r'}\cap K_{t})$.  From the definition of clique paths it can be inferred that $M$ is disjoint from $K_{t'}$, and thus $M$ does not contain any exposed vertex of $S$.  On the one hand, no splitter of $S$ is adjacent to $M$; on the other hand, there is no splitter of $M$ in $S$.  Thus, $M$ is a module of $G$, contradicting the assumption.
   Therefore,  $K_{r'}$ is after $K_{t'}$ in $\mathcal{K}'$.  Again, by Lemma~\ref{lem:modules}, there is a module $M$ of $G[S\cup X]$ that contains all vertices in $\bigcup_{j=t}^{t'} K_{j}\setminus (K_{t}\cap K_{t'})$.  Note that both $v_p$ and $v_z$ are in $M$, and $X\subseteq K_{t}\cap K_{t'}$.
  If $G[M]$ is connected or if $X\subset K_{t}\cap K_{t'}$, then an \textsc{lbfs} ordering of $G[S]$ from $v_{z}$ cannot end at $v_{p}$: Before vertices in $M$ are exhausted, there is a vertex in $M$ whose label is a proper superset of $K_{t}\cap K_{t'}$, while the  label of any vertex in $(S\cup X)\setminus M$ is a subset of $K_{t}\cap K_{t'}$.  Therefore, $K_{t}\setminus K_{t'}= K_{t}\setminus X$ and $K_{t'}\setminus K_{t} = K_{t'}\setminus X$, and they belong to different components of $G[S]$.  Moreover, $K_{r'}\setminus X$ is in another component of $G[S]$.  Then the components of $G[S]$ containing $K_{r'}\setminus X$ and $K_{t}\setminus X$ form a module $M'$ of $G$, and $M'$ is disjoint from $K_{t'}$.  Thus, $M'$ is a module of $G$.  But this contradicts our assumption.  The proof is now complete.
\end{proof}

It is not difficult to prove the following result.  Since we are not using it, we omit the proof.
\begin{remark}
  Let $\sigma$ be an \textsc{lbfs}$^{+}$ ordering of an interval graph $G$, and let $M$ be a major module of $G$.  If $M$ does not contain the first vertex of $\sigma$, then vertices in $M$ appear consecutive in $\sigma$.
\end{remark}

As a final remark, Theorem~\ref{thm:uig} also implies Theorem 4.30 of \cite{li-14-lbfs-interval-recognition}, namely, the algorithm for interval recognition always returns a correct umbrella ordering for a unit interval graph.

\section{Implementation and concluding remarks}\label{sec:conclusion}

All the procedures are implemented using the idea of partition refinement \cite{habib-00-LBFS-and-partition-refinement}.  We sketch here the steps very briefly.  Similar as \textsc{lbfs}$^{+}$, we can start \textsc{lbfs}$^{\delta}$ with the vertices ordered by their degrees.  The procedure \textsc{lbfs}$^{\uparrow}$ is more complicated.  Recall that at the beginning we renumber the vertices according to $\tau^{+}$.  We maintain an array $d$ that is initialized as $d[i] = |\{v_j\in N(v_i)\mid j < i\}|$; i.e., $d[i]$ is the number of neighbors of $v_i$ that are before $v_i$ in $\tau^{+}$.  We start the partition procedure with the vertices sorted by $\max\{j\mid v_j\in N[v_i]\}$, and for vertices with the same value, sort them in the reversal of their indices.  When a vertex $v_i$ is visited, we decrease $d[j]$ for each unvisited neighbor $v_j\in N(v_i)$.  Then condition of step~2.5 of \textsc{lbfs}$^{\uparrow}$ is satisfied if and only if $d(p) > 0$, and for both steps~2.6 and 2.7, it suffices to take the last vertex in the list for the current snapshot.

Among the known recognition algorithms for interval graphs, the ones by Hsu and Ma~\cite{hsu-99-recognizing-interval-graphs} and Li and Wu~\cite{li-14-lbfs-interval-recognition} are arguably the simplest.  However, they are significantly more complicated than the algorithms of Rose et al.~\cite{rose-76-vertex-elimination} for chordal graphs, not to mention the simpler one in Tarjan and Yannakakis~\cite{tarjan-84-chordal-recognition}.  Since interval graphs are conceptually simpler than chordal graphs, it may not be safe to call either of them the ultimate algorithm for the recognition of interval graphs.  On the other hand, we believe that they are close to the ultimate algorithm, if such an algorithm does exist.  Toward this direction, one step might be better understanding the well-anchored orderings of an interval graph.  In particular, can we produce one with only one or two sweeps of graph searches?

\paragraph{Acknowledgment.}
This work was based on the lecture notes for the \emph{CCF (China Computer Federation) Summer School on Algorithmic Graph Theory}, conducted in August 2020.  I thank Xiaoming Sun for inviting me to teach the summer school, and North Minzu University (through Xiaofeng Wang) for the finance support for the summer school.


\begin{thebibliography}{10}

\bibitem{booth-76-pq-tree}
Kellogg~S. Booth and George~S. Lueker.
\newblock Testing for the consecutive ones property, interval graphs, and graph
  planarity using {$P Q$}-tree algorithms.
\newblock {\em Journal of Computer and System Sciences}, 13(3):335--379, 1976.
\newblock A preliminary version appeared in STOC 1975.
\newblock \href {https://doi.org/10.1016/S0022-0000(76)80045-1}
  {\path{doi:10.1016/S0022-0000(76)80045-1}}.

\bibitem{cao-19-end-vertices}
Yixin Cao, Zhifeng Wang, Guozhen Rong, and Jianxin Wang.
\newblock Graph searches and their end vertices.
\newblock In Pinyan Lu and Guochuan Zhang, editors, {\em Proceedings of the
  30th International Symposium on Algorithms and Computation (ISAAC)}, volume
  149 of {\em LIPIcs}, pages 1:1--1:18. Schloss Dagstuhl - Leibniz-Zentrum
  f{\"{u}}r Informatik, 2019.
\newblock \href {https://doi.org/10.4230/LIPIcs.ISAAC.2019.1}
  {\path{doi:10.4230/LIPIcs.ISAAC.2019.1}}.

\bibitem{corneil-04-survey-lbfs}
Derek~G. Corneil.
\newblock Lexicographic breadth first search - {A} survey.
\newblock In Juraj Hromkovic, Manfred Nagl, and Bernhard Westfechtel, editors,
  {\em Graph-Theoretic Concepts in Computer Science (WG)}, volume 3353 of {\em
  LNCS}, pages 1--19. Springer, 2004.
\newblock \href {https://doi.org/10.1007/978-3-540-30559-0_1}
  {\path{doi:10.1007/978-3-540-30559-0_1}}.

\bibitem{corneil-04-recognize-uig}
Derek~G. Corneil.
\newblock A simple 3-sweep {LBFS} algorithm for the recognition of unit
  interval graphs.
\newblock {\em Discrete Applied Mathematics}, 138(3):371--379, 2004.
\newblock \href {https://doi.org/10.1016/j.dam.2003.07.001}
  {\path{doi:10.1016/j.dam.2003.07.001}}.

\bibitem{corneil-95-recognition-unit-interval}
Derek~G. Corneil, Hiryoung Kim, Sridhar Natarajan, Stephan Olariu, and Alan~P.
  Sprague.
\newblock Simple linear time recognition of unit interval graphs.
\newblock {\em Information Processing Letters}, 55(2):99--104, 1995.
\newblock \href {https://doi.org/10.1016/0020-0190(95)00046-F}
  {\path{doi:10.1016/0020-0190(95)00046-F}}.

\bibitem{corneil-10-end-vertices-lbfs}
Derek~G. Corneil, Ekkehard K{\"o}hler, and Jean{-}Marc Lanlignel.
\newblock On end-vertices of lexicographic breadth first searches.
\newblock {\em Discrete Applied Mathematics}, 158(5):434--443, 2010.
\newblock \href {https://doi.org/10.1016/j.dam.2009.10.001}
  {\path{doi:10.1016/j.dam.2009.10.001}}.

\bibitem{corneil-98-interval-recognition}
Derek~G. Corneil, Stephan Olariu, and Lorna Stewart.
\newblock The ultimate interval graph recognition algorithm? (extended
  abstract).
\newblock In Howard~J. Karloff, editor, {\em Proceedings of the 9th Annual
  ACM-SIAM Symposium on Discrete Algorithms (SODA)}, pages 175--180.
  {ACM/SIAM}, 1998.

\bibitem{corneil-09-lbfs-strucuture-and-interval-recognition}
Derek~G. Corneil, Stephan Olariu, and Lorna Stewart.
\newblock The {LBFS} structure and recognition of interval graphs.
\newblock {\em SIAM Journal on Discrete Mathematics}, 23(4):1905--1953, 2009.
\newblock A preliminary version appeared in SODA 1998.
\newblock \href {https://doi.org/10.1137/S0895480100373455}
  {\path{doi:10.1137/S0895480100373455}}.

\bibitem{deng-96-proper-interval-and-cag}
Xiaotie Deng, Pavol Hell, and Jing Huang.
\newblock Linear-time representation algorithms for proper circular-arc graphs
  and proper interval graphs.
\newblock {\em SIAM Journal on Computing}, 25(2):390--403, 1996.
\newblock \href {https://doi.org/10.1137/S0097539792269095}
  {\path{doi:10.1137/S0097539792269095}}.

\bibitem{dirac-61-chordal-graphs}
Gabriel~A. Dirac.
\newblock On rigid circuit graphs.
\newblock {\em Abhandlungen aus dem Mathematischen Seminar der Universit{\"a}t
  Hamburg}, 25(1):71--76, 1961.
\newblock \href {https://doi.org/10.1007/BF02992776}
  {\path{doi:10.1007/BF02992776}}.

\bibitem{fulkerson-65-interval-graphs}
Delbert~R. Fulkerson and Oliver~A. Gross.
\newblock Incidence matrices and interval graphs.
\newblock {\em Pacific Journal of Mathematics}, 15(3):835--855, 1965.
\newblock \href {https://doi.org/10.2140/pjm.1965.15.835}
  {\path{doi:10.2140/pjm.1965.15.835}}.

\bibitem{gimbel-88-end-vertices}
John Gimbel.
\newblock End vertices in interval graphs.
\newblock {\em Discrete Applied Mathematics}, 21(3):257--259, 1988.
\newblock \href {https://doi.org/10.1016/0166-218X(88)90071-6}
  {\path{doi:10.1016/0166-218X(88)90071-6}}.

\bibitem{habib-00-LBFS-and-partition-refinement}
Michel Habib, Ross~M. McConnell, Christophe Paul, and Laurent Viennot.
\newblock Lex-{BFS} and partition refinement, with applications to transitive
  orientation, interval graph recognition and consecutive ones testing.
\newblock {\em Theoretical Computer Science}, 234(1-2):59--84, 2000.
\newblock \href {https://doi.org/10.1016/S0304-3975(97)00241-7}
  {\path{doi:10.1016/S0304-3975(97)00241-7}}.

\bibitem{hsu-95-recognition-cag}
Wen-Lian Hsu.
\newblock {$O(m\cdot n)$} algorithms for the recognition and isomorphism
  problems on circular-arc graphs.
\newblock {\em SIAM Journal on Computing}, 24(3):411--439, 1995.
\newblock \href {https://doi.org/10.1137/S0097539793260726}
  {\path{doi:10.1137/S0097539793260726}}.

\bibitem{hsu-99-recognizing-interval-graphs}
Wen-Lian Hsu and Tze-Heng Ma.
\newblock Fast and simple algorithms for recognizing chordal comparability
  graphs and interval graphs.
\newblock {\em SIAM Journal on Computing}, 28(3):1004--1020, 1999.
\newblock \href {https://doi.org/10.1137/S0097539792224814}
  {\path{doi:10.1137/S0097539792224814}}.

\bibitem{hsu-03-pc-trees}
Wen-Lian Hsu and Ross~M. McConnell.
\newblock {PC} trees and circular-ones arrangements.
\newblock {\em Theoretical Computer Science}, 296(1):99--116, 2003.
\newblock \href {https://doi.org/10.1016/S0304-3975(02)00435-8}
  {\path{doi:10.1016/S0304-3975(02)00435-8}}.

\bibitem{korte-89-recognizing-interval-graphs}
Norbert Korte and Rolf~H. M{\"o}hring.
\newblock An incremental linear-time algorithm for recognizing interval graphs.
\newblock {\em SIAM Journal on Computing}, 18(1):68--81, 1989.
\newblock \href {https://doi.org/10.1137/0218005} {\path{doi:10.1137/0218005}}.

\bibitem{lekkerkerker-62-interval-graphs}
Cornelis~G. Lekkerkerker and J.~Ch. Boland.
\newblock Representation of a finite graph by a set of intervals on the real
  line.
\newblock {\em Fundamenta Mathematicae}, 51:45--64, 1962.
\newblock \href {https://doi.org/10.4064/fm-51-1-45-64}
  {\path{doi:10.4064/fm-51-1-45-64}}.

\bibitem{li-14-lbfs-interval-recognition}
Peng Li and Yaokun Wu.
\newblock A four-sweep {LBFS} recognition algorithm for interval graphs.
\newblock {\em Discrete Mathematics {\&} Theoretical Computer Science},
  16(3):23--50, 2014.

\bibitem{looges-93-greedy-algorithms-uig}
Peter~J. Looges and Stephan Olariu.
\newblock Optimal greedy algorithms for indifference graphs.
\newblock {\em Computers \& Mathematics with Applications}, 25(7):15--25, 1993.
\newblock \href {https://doi.org/10.1016/0898-1221(93)90308-I}
  {\path{doi:10.1016/0898-1221(93)90308-I}}.

\bibitem{panda-09-bicompatible-elimination-ordering}
B.~S. Panda and Sajal~K. Das.
\newblock A parallel algorithm for generating bicompatible elimination
  orderings of proper interval graphs.
\newblock {\em Information Processing Letters}, 109(18):1041--1046, 2009.
\newblock \href {https://doi.org/10.1016/j.ipl.2009.06.011}
  {\path{doi:10.1016/j.ipl.2009.06.011}}.

\bibitem{ramalingam-88-domination-interval-graphs}
G.~Ramalingam and C.~Pandu Rangan.
\newblock A unified approach to domination problems on interval graphs.
\newblock {\em Information Processing Letters}, 27(5):271--274, 1988.
\newblock \href {https://doi.org/10.1016/0020-0190(88)90091-9}
  {\path{doi:10.1016/0020-0190(88)90091-9}}.

\bibitem{roberts-69-indifference-graphs}
Fred~S. Roberts.
\newblock Indifference graphs.
\newblock In Frank Harary, editor, {\em Proof Techniques in Graph Theory (Proc.
  Second Ann Arbor Graph Theory Conf., 1968)}, pages 139--146. Academic Press,
  New York, 1969.

\bibitem{rose-76-vertex-elimination}
Donald~J. Rose, Robert~Endre Tarjan, and George~S. Lueker.
\newblock Algorithmic aspects of vertex elimination on graphs.
\newblock {\em SIAM Journal on Computing}, 5(2):266--283, 1976.
\newblock A preliminary version appeared in STOC 1975.
\newblock \href {https://doi.org/10.1137/0205021} {\path{doi:10.1137/0205021}}.

\bibitem{simon-91-interval}
Klaus Simon.
\newblock A new simple linear algorithm to recognize interval graphs.
\newblock In {\em Computational Geometry - Methods, Algorithms and
  Applications, International Workshop on Computational Geometry CG'91, Bern,
  Switzerland, March 21-22, 1991}, pages 289--308, 1991.
\newblock \href {https://doi.org/10.1007/3-540-54891-2_22}
  {\path{doi:10.1007/3-540-54891-2_22}}.

\bibitem{tarjan-84-chordal-recognition}
Robert~Endre Tarjan and Mihalis Yannakakis.
\newblock Simple linear-time algorithms to test chordality of graphs, test
  acyclicity of hypergraphs, and selectively reduce acyclic hypergraphs.
\newblock {\em SIAM Journal on Computing}, 13(3):566--579, 1984.
\newblock With Addendum in the same journal, 14(1):254-255, 1985.
\newblock \href {https://doi.org/10.1137/0213035} {\path{doi:10.1137/0213035}}.

\bibitem{tedder-08}
Marc Tedder, Derek~G. Corneil, Michel Habib, and Christophe Paul.
\newblock Simpler linear-time modular decomposition via recursive factorizing
  permutations.
\newblock In {\em Automata, Languages and Programming (ICALP)}, volume 5125 of
  {\em LNCS}, pages 634--645, Berlin Heidelberg, 2008. Springer-Verlag.
\newblock \href {https://doi.org/10.1007/978-3-540-70575-8_52}
  {\path{doi:10.1007/978-3-540-70575-8_52}}.

\end{thebibliography}
\end{document}